\documentclass[runningheads,envcountsame]{llncs}
\usepackage[T1]{fontenc}
\usepackage{lmodern}

\usepackage{amsmath}
\usepackage{amsthm}
\usepackage{amssymb}
\usepackage{xcolor}
\definecolor{refcol}{RGB}{0,51,153}
\definecolor{citecol}{RGB}{44,160,46}
\definecolor{urlcol}{RGB}{0,51,153} 
\usepackage[hidelinks,colorlinks,linkcolor=refcol,citecolor=citecol,urlcolor=urlcol
  ]{hyperref}
\usepackage{cleveref}
\usepackage{stmaryrd}
\usepackage[createShortEnv]{proof-at-the-end}

\usepackage{IEEEtrantools}

\newif\iflongversion\longversiontrue  
\newif\ifdebugversion\debugversiontrue  

\longversiontrue 
\debugversionfalse

\iflongversion
\newcommand{\citelongversion}[1]{\Cref{#1}\xspace}
\else
\newcommand{\citelongversion}[1]{the extended version \cite{DBLP:conf/ijcar/AbouElWafaP26}\xspace}
\pratendSetGlobal{no link to proof}
\fi

\usepackage{paralist}
\newenvironment*{caselist}[1][- Case]{%
    \newcommand{\caseof}[1]{\item[\noindent\emph{#1 ##1}:]}
    \begin{inparaitem}
        }{%
    \end{inparaitem}
}

\usepackage{enumitem}

\usepackage{math}
\usepackage[prefixflatinterpret,bracketmodalinterpret,fixformat,silentconst,sidenotecalculus,longseqcontext,boldmquant]{logic}
\usepackage[prefixflatinterpret,bracketmodalinterpret,fixformat,silentconst,differentialdL,simplenames]{dL}%

\definecolor{darkishgray}{rgb}{.35,.35,.35}
\renewcommand*{\irrulename}[1]{\text{\textcolor{darkishgray}{\upshape#1}}}
\renewcommand{\linferenceRuleLabelPrefix}{}
\renewcommand{\linferenceRuleNameSeparation}{\hspace{1em}}
\renewcommand{\linferPremissSeparation}{\hspace{0.4cm}}
 
\renewcommand{\ptest}[1]{{?#1}}
 
\definecolor{todored}{RGB}{200,60,60}     
\definecolor{todoblue}{RGB}{60,100,180}   
\definecolor{remembergreen}{RGB}{40,140,90}
\definecolor{todogray}{RGB}{140,140,140}

\definecolor{highlight}{RGB}{0,0,0}  
\definecolor{annotgrey}{gray}{0.25}

\ifdebugversion
\definecolor{pos}{RGB}{0,114,178}      % deep blue
\definecolor{neg}{RGB}{213,94,0}      % vermillion (orange-red)
\definecolor{rpos}{RGB}{0,158,115}    % bluish green
\definecolor{rneg}{RGB}{204,121,167}  % soft magenta
\newcommand{\todo}[1]{\footnote{\textcolor{todored}{Todo: #1}}}
\else
\definecolor{pos}{RGB}{0,0,0}     % deep blue
\definecolor{neg}{RGB}{0,0,0}      % vermillion (orange-red)
\definecolor{rpos}{RGB}{0,0,0}    % bluish green
\definecolor{rneg}{RGB}{0,0,0}  % soft magenta
\newcommand{\todo}[1]{}
\fi
\definecolor{setrelation}{RGB}{0, 0, 0}  
 
\newcommand{\highlight}[1]{{\color{highlight}#1}}

\newcommand{\fml}{\varphi}
\newcommand{\fmlb}{\rho}

\newcommand{\prog}{\alpha}
\newcommand{\progb}{\gamma}

\newcommand{\strictsymb}{\color{annotgrey}\circ}
\newcommand{\laxsymb}{\color{annotgrey}\bullet}
\newcommand{\rstrictsymb}{\color{annotgrey}\diamond}
\newcommand{\rlaxsymb}{\color{annotgrey}\rotatebox{45}{\scalebox{0.3}{$\blacksquare$}}}

\newcommand{\posof}[1]{{\color{pos}#1}^{\scriptscriptstyle{\color{pos}\kern-1pt\strictsymb}}}
\newcommand{\negof}[1]{{\color{neg}#1}^{\scriptscriptstyle{\color{neg}\kern-1pt\laxsymb}}}

\newcommand{\posrof}[1]{{\color{rpos}#1}^{\scriptscriptstyle{\color{rpos}\kern-1pt\rstrictsymb}}}
\newcommand{\negrof}[1]{{\color{rneg}#1}^{\scriptscriptstyle{\color{rneg}\kern-1pt\rlaxsymb}}}

\newcommand{\posfml}{\posof{\fml}}
\newcommand{\posfmlb}{\posof{\fmlb}}
\newcommand{\negfml}{\negof{\fml}}
\newcommand{\negfmlb}{\negof{\fmlb}}

\newcommand{\posprog}{\posof{\prog}}
\newcommand{\posprogb}{\posof{\progb}}
\newcommand{\negprog}{\negof{\prog}}
\newcommand{\negprogb}{\negof{\progb}}

\newcommand{\posrfml}{\posrof{\fml}}
\newcommand{\posrfmlb}{\posrof{\fmlb}}
\newcommand{\negrfml}{\negrof{\fml}}
\newcommand{\negrfmlb}{\negrof{\fmlb}}

\newcommand{\posrprog}{\posrof{\prog}}
\newcommand{\posrprogb}{\posrof{\progb}}
\newcommand{\negrprog}{\negrof{\prog}}
\newcommand{\negrprogb}{\negrof{\progb}}

\newcommand{\term}{v}
\newcommand{\termb}{w}

\newcommand{\ivar}{x}
\newcommand{\ivarb}{y}
\newcommand{\ivarc}{z}

\newcommand{\ival}{a}

\newcommand{\strictfml}{strict\xspace}
\newcommand{\Strictfml}{Strict\xspace}
\newcommand{\laxfml}{non-strict\xspace}
\newcommand{\Laxfml}{Non-strict\xspace}
\newcommand{\exactprg}{strict\xspace}
\newcommand{\approxprog}{non-strict\xspace}

\newcommand{\sem}[1]{{\llbracket}#1{\rrbracket}}

\newcommand{\semopen}[1]{{\color{pos}\llbracket}#1{\color{pos}\rrbracket}}
\newcommand{\semclosed}[1]{{\color{neg}\llbracket}#1{\color{neg}\rrbracket}}

\newcommand{\semexactwr}[2]{\awinningregion{{\color{pos}\llbracket}#1{\color{pos}\rrbracket}}{#2}}
\newcommand{\semapproxwr}[2]{\awinningregion{{\color{neg}\llbracket}#1{\color{neg}\rrbracket}}{#2}}

\newcommand{\semwr}[2]{\awinningregion{{\llbracket}#1{\rrbracket}}{#2}}

\newcommand{\semwrp}[2]{\semwr{#1}{(#2)}}
\newcommand{\semexactwrp}[2]{\semexactwr{#1}{(#2)}}
\newcommand{\semapproxwrp}[2]{\semapproxwr{#1}{(#2)}}

\newcommand{\semexactrepeat}[4][\compactstateset]{{{\color{pos}\llbracket}#2{\color{pos}\rrbracket}^{#4}}{#3}}
\newcommand{\semexactrepeatp}[4][\compactstateset]{\semexactrepeat[#1]{#2}{(#3)}{#4}}

\newcommand{\semterm}[2][\state]{#1\llbracket#2\rrbracket}

\newcommand{\union}{\cup}
\newcommand{\intersection}{\cap}
\renewcommand{\complement}[1]{#1^\mathrm{c}}

\newcommand{\states}{\mathcal{S}}
\newcommand{\state}{\omega}
\newcommand{\stateb}{\nu}

\newcommand{\staterepvarbyval}[3][\state]{#1_{#2}^{#3}}
\newcommand{\staterepvarsbyvals}[3][\state]{#1_{#2}^{#3}}

\newcommand{\openset}{U}
\newcommand{\opensetb}{V}
\newcommand{\opensetc}{W}

\newcommand{\closedset}{C}

\newcommand{\rel}{R}
\newcommand{\relb}{S}
\newcommand{\transitiveclosure}[1]{#1^*}
\newcommand{\awinningregion}[2]{{#1}^{-1}#2}
\newcommand{\awinningregionp}[2]{\awinningregion{#1}{(#2)}}

\newcommand{\topspace}{X}
\newcommand{\topspaceb}{Y}
\newcommand{\topspaceelement}{x}
\newcommand{\topspaceelementb}{y}
\newcommand{\topspaceelementc}{z}

\newcommand{\variables}{\mathcal{V}}

\newcommand{\ratconst}{q}

\newcommand{\rdL}{\ensuremath{\dL_{\mathsf{r}}}\xspace}
\newcommand{\rrdL}{\ensuremath{\dL_{\mathsf{r}\diamond}}\xspace}

\newcommand{\myvect}[1]{%
    \mathsf{#1}
}
\newcommand{\ode}[2]{\D{#1}=#2(#1)} 

\newcommand{\synodevec}[2]{\D{\myvect{#1}}={\myvect{#2}}}
\newcommand{\synodevecexpl}[3][n]{\D{#2_1}={#3_1}\synodecomp\ldots\synodecomp\D{#2}_#1={#3_#1}}
\newcommand{\synodecomp}{,}
\newcommand{\flowdomain}[1]{\Omega_{#1}}
\newcommand{\odeflowfunc}[1]{\Phi_{#1}}
\newcommand{\odeflow}[3]{\odeflowfunc{#1}(#2,#3)}

\newcommand{\lforallbdd}[2][]{\forall_{#1}{#2}\;}
\newcommand{\lexistsbdd}[2][]{\exists_{#1}{#2}\;}

\newcommand{\syninterv}{I}

\newcommand{\timevar}{\tau}

\newcommand{\mtime}{t}
\newcommand{\mtimeb}{s}
\newcommand{\evolutionboundterm}{k}
\newcommand{\evolutionboundtermsem}{K}

\newcommand{\lipsch}{\ell}
\newcommand{\fbdd}{m}
\newcommand{\mlipsch}{L}
\newcommand{\mfbdd}{M}
\newcommand{\iterationprogabb}{\posrof{\mathrm{step}}}
\newcommand{\secondderbound}{\posrof{\mathrm{bound}}}

\newcommand{\staterestr}[2][\state]{#1{\upharpoonright}#2}

\newcommand{\vecfield}{f}
\newcommand{\syntvectorfield}[2]{\vecfield_{#2}}

\newcommand{\stateevolve}[2]{\stateevolvefunc{#1}(#2)}
\newcommand{\stateevolvefunc}[1]{\Psi_{#1}}

\newcommand{\rankof}[1]{\mathrm{rank}(#1)}

\newcommand{\provrel}{\prec}
\newcommand{\provrelhalf}{\preccurlyeq}
\newcommand{\provsetrel}{\mathrel{\color{setrelation}\prec}}
\newcommand{\provsetrelhalf}{\mathrel{\color{setrelation}\preccurlyeq}}

\newcommand{\setoffmlsets}{\Sigma}
\newcommand{\fmlset}{\Gamma}
\newcommand{\fmlsetb}{\Delta}

\newcommand{\fmlreplacevarby}[3][\fml]{#1_{#2}^{#3}}
\newcommand{\termreplacevarby}[3][term]{#1_{#2}^{#3}}
\newcommand{\fmlsafereplacevarby}[3][term]{#1\tfrac{#3}{#2}}

\newcommand{\epsint}[2][\varepsilon]{#2_{\myvect{\ivar}\pm#1}}

\newcommand{\sgandalfsymb}{\scalebox{0.8}{$\overleftrightarrow{\,\mathrm{\Delta}\,}$}\xspace}

\newcommand{\syntnorm}[1]{{\abs{{#1}}}}

\newcommand{\valid}[1]{\vDash#1}
\newcommand{\prov}[1]{\vdash#1}

\newcommand{\lnotp}[1]{\lnot(#1)}

\newcommand{\eqcref}[2]{\Cref{#1}\eqref{#2}}

\newcommand{\piter}[2]{#1^{\leq#2}}

\newcommand{\robnot}[1]{{\sim}{}#1}
\newcommand{\robnotp}[1]{\robnot{(#1)}}

\newcommand{\wellposed}{robustly determined\xspace}

\renewcommand{\humod}[2]{#1{\eqdef}#2}
 
\AtBeginDocument{
  \newsavebox{\Rval}%
  \sbox{\Rval}{$\mathbb{R}$}
}

\begin{document}
\title{Complete Robust Hybrid Systems Reachability}
\author{Noah {Abou El Wafa}\orcidID{0000-0002-3987-9919} \and
Andr{\'{e}} Platzer\orcidID{0000-0001-7238-5710}}
\institute{Karlsruhe Institute of Technology, Karlsruhe, Germany
  \email{\{noah.abouelwafa,platzer\}@kit.edu}}
\authorrunning{Noah {Abou El Wafa} and Andr{\'{e}} Platzer}
\maketitle
\begin{abstract}
  This paper introduces \emph{robust differential dynamic logic} (a fragment of differential dynamic logic) to specify and reason about \emph{robust hybrid systems}.
  By small, natural, and practically meaningful syntactic restrictions, specifications are ensured (by construction) to be topologically open and, thus, robust with respect to infinitesimal perturbations without explicit quantitative margins of error in the syntax or in proofs.
  The main result is a proof of \emph{absolute completeness} of robust differential dynamic logic for reachability properties of general hybrid systems.
  The proof is constructive, self-contained, and demonstrates how robustly-correct hybrid systems reachability specifications can be automatically verified through proof.
\end{abstract}
% 
% 
%%%%%%%%%%%%%%%%%%%%%%%%%%%%%

\section{Introduction}
\label{sec:introduction}

The study and verification of hybrid systems properties is of fundamental importance for many cyber-physical systems applications.
Yet, even reachability is decidable \emph{only} for the most restrictive classes of hybrid systems and the simplest properties \cite{DBLP:conf/stoc/HenzingerKPV95,DBLP:conf/concur/CassezL00}.
The theoretical hardness is often attributed to \emph{excessive precision} that does not reflect physical reality, which in practice always involves some error or imprecision \cite{DBLP:conf/csl/Franzle99}.
This has led to a bifurcation of approaches to hybrid systems verification:
On the one hand, there are approaches that \emph{weaken the semantics} to regain decidability by admitting positive \(\delta>0\) errors in the decisions, for example \(\delta\)-decidability \cite{DBLP:conf/tacas/KongGCC15} or interval approximation \cite{DBLP:journals/fmsd/Ratschan14}.
On the other hand, precise and deductive approaches, such as differential dynamic logic \cite{DBLP:journals/jar/Platzer08}, retain exact semantics, but rely on \emph{undecidable oracles} \cite{DBLP:conf/lics/Platzer12a} or other restrictions \cite{DBLP:journals/jacm/PlatzerT20} for completeness.
But neither of these approaches is at the same time free of perturbations or errors \emph{and} complete in general.

This paper bridges the gap, retaining \emph{exact semantics} and nonetheless achieving \emph{completeness}, by focusing on realistic \emph{robust properties} of hybrid systems in a logical calculus.
Importantly, robustness is a topological property (openness) that does not require any explicit perturbations, margins of error, or positive disturbances \(\delta>0\).
This notion of \emph{infinitesimal robustness} is captured purely syntactically within differential dynamic logic \dL \cite{DBLP:journals/jar/Platzer08,DBLP:journals/jar/Platzer17}, a highly expressive, powerful logic, which extends first-order dynamic logic with differential equations to reason about continuous dynamics.
The \emph{robust} fragment (robust differential dynamic logic) syntactically restricts boundary phenomena, which demand unnecessary precision that cannot be delivered symbolically or computationally.
This restriction is simple, intuitive and modest, as it demands only that in positive positions (within an even number of negations) all inequalities are strict, and in negative positions (within an odd number of negations) all inequalities are weak.
Intuitively, any claim to be proven should be robustly true (i.e., strict or open), while assumptions should conservatively be considered marginal (non-strict or closed).
Little of the \emph{practical} expressiveness of \dL is lost, since the difference is only between \(x>0\) and \(x\geq 0\), which is not practically relevant for the verification of cyber-physical systems.
Robust \dL provides a clean, simple, syntactic way of describing and reasoning deductively about general robust hybrid systems properties without any need for manual error management.

The \emph{large expressiveness} of \dL means that it is impossible to prove all valid formulas in a sound and effective proof calculus \cite{DBLP:conf/lics/Platzer12a} or even to decide the truth of all reachability problems formulated in \dL (a consequence of Matiyasevich's theorem).
It is shown that robustness solves this problem by reducing the actual expressiveness to a completely axiomatizable level without sacrificing practical expressiveness.
So robustness \emph{uniformly} handles all difficulties that arise when proving a valid reachability formula.
For example, valid existential quantifiers can be correctly instantiated by rational values and unbounded reachability can be bounded by compactness.
Unlike some previous completeness theorems \cite{DBLP:journals/jar/Platzer08,DBLP:conf/lics/Platzer12a,DBLP:journals/jar/Platzer17}, completeness for robust hybrid systems reachability is \emph{not relative} and does not require an undecidable oracle: any true robust hybrid systems reachability property is provable unconditionally.
Importantly, reachability properties are not time-bounded, but can involve hybrid systems running for unbounded lengths of time and repeating loops for an arbitrary number of iterations.

\subsubsection{Contributions.}
First, \emph{robust differential dynamic logic} is introduced as a complete fragment of \dL capable of expressing all robust properties of hybrid systems.
It is shown that the semantics naturally reflects this robustness through a topological characterization (\Cref{sec:doandrdl}).
Second, a simple iterative axiomatization of ODE properties is presented (\Cref{sec:calculus}).
Third, robustness is used to provide a constructive completeness proof for robust hybrid systems reachability properties.
This \emph{non-relative completeness proof} for a general class of hybrid systems properties overcomes the previous reliance on undecidable oracles and restrictions to algebraic hybrid systems (\Cref{sec:completeness}).

\iflongversion
\else
Full proofs of all results can be found in \citelongversion{ }.
\fi

\subsubsection{Related Work.}
Hybrid systems verification is undecidable \cite{DBLP:conf/stoc/HenzingerKPV95}.
A \(\delta\)-decidability algorithm handles logical reachability analysis via SMT solving approximately for accuracy \(\delta>0\) \cite{DBLP:conf/tacas/KongGCC15}.
Similar syntactic perturbations have been used to prove quasi-decidability of safety for hybrid systems \cite{DBLP:journals/fmsd/Ratschan14}.
Reachability for restricted robust systems by approximation has been proved decidable \cite{DBLP:conf/csl/Franzle99}.
Unlike these approaches, robust differential dynamic logic handles robustness in a simple and topological way, without requiring explicit approximations.

Completeness of \dL relative to its continuous fragment and its discrete fragment has been shown \cite{DBLP:conf/lics/Platzer12a}. Absolute completeness was established for algebraic properties of algebraic hybrid systems \cite{DBLP:journals/jacm/PlatzerT20} and invariance properties of hybrid systems \cite{DBLP:journals/jacm/PlatzerT20,DBLP:conf/adhs/TanP21}.
This paper shows \emph{absolute} completeness for robust reachability properties of \emph{general hybrid systems}, extending the completeness of open first-order properties of purely continuous dynamical systems \cite{DBLP:journals/jacm/PlatzerQ25}.

Conceptually related is the proof of completeness of the uninterpreted diamond fragment of first-order dynamic logic \cite{DBLP:journals/iandc/Schmitt84}.
The connection between open sets of reals and decidability is fundamental in computable analysis \cite{DBLP:series/txtcs/Weihrauch00}.

\section{Differential Dynamic Logic and Robustness}
\label{sec:doandrdl}

This section first recalls the syntax and semantics of differential dynamic logic \dL \cite{DBLP:journals/jar/Platzer08} (\Cref{sec:syntaxofdl,sec:dlsemantics}).
Robust differential dynamic logic \rdL and its reachability fragment \rrdL are then defined as syntactic fragments of \dL (\Cref{sec:robustdl,sec:robustreachdl}).
The semantics of \rrdL are shown to be robust under infinitesimal perturbations and thus capture the \emph{robustly true} formulas (\Cref{sec:robustnesssem}).
 
\subsection{Syntax of Differential Dynamic Logic}\label{sec:syntaxofdl}
 
Formally, \dL extends first-order dynamic logic for (discrete) programs, with continuous programs to model dynamics of hybrid systems.
As \dL is interpreted over the real numbers, the terms of \dL are ordinary terms of real arithmetic:
\begin{align*}
  \term,\termb & ::=\ivar \mid \ratconst\mid\term+\termb \mid \term\cdot\termb
               &
  \text{\textcolor{annotgrey}{\dL terms}}\hspace{-10em}
\end{align*}
where \(\ivar\in\variables\) is a first-order variable from some fixed set \(\variables\) of variables and \(\ratconst\) is a rational constant.
Terms are essentially syntactic descriptions of polynomials in the variables \(\variables\) with rational coefficients (\(\rationals[\variables]\)).
Symbols \(\myvect{\ivar}\) and \(\myvect{\term}\) implicitly range over \emph{tuples} of variables \(\ivar_1,\ldots,\ivar_n\) and terms \(\term_1,\ldots,\term_n\), respectively.

Formulas \(\fml\) and hybrid programs \(\prog\) of \dL are described by:
\begin{align*}
  \fml,\fmlb   & ::= \term >\termb \mid \lnot\fml \mid \fml\lor\fmlb\mid \lexists{\ivar}\fml \mid \ddiamond{\prog} \fml
               &
  \text{\textcolor{annotgrey}{\dL formulas}}
  \\
  \prog,\progb & ::= \humod{\ivar}{\term} \mid \ptest{\fml} \mid \pchoice{\prog}{\progb} \mid \prog;\progb \mid \prepeat{\prog} \mid \pevolvein{\synodevec{\ivar}{\term}}{\fml}
               &
  \text{\textcolor{annotgrey}{\dL programs}}
\end{align*}
where \(\myvect{\ivar}\) and \(\myvect{\term}\) are tuples (of the same length) of variables and terms, respectively, so that \(\pevolvein{\synodevec{\ivar}{\term}}{\fml}\) describes the symbolic ordinary differential equation \(\synodevecexpl[n]{\ivar}{\term}\) on the domain defined by the formula \(\fml\).

For convenience it is assumed, without loss of generality, that there is a special variable \(\tau\) and every continuous evolution modality contains \(\D{\tau}=1\).
(This does not change the proof theory by the differential ghost axiom \cite{DBLP:journals/jacm/PlatzerT20}.)

\subsubsection{Interpretation of \dL.}

The modalities \(\ddiamond{\prog}\fml\) of \dL are parameterized by \emph{hybrid programs} \(\prog\), which describe hybrid systems \emph{syntactically}.
A formula of the form \(\ddiamond{\prog}\fml\) asserts the existence of an execution of the hybrid program \(\prog\) that reaches a state in which \(\fml\) is true.
A hybrid program \(\prog\) describes a non-deterministic hybrid system.
The \emph{deterministic assignment} \(\humod{\ivar}{\term}\) models a simple state change, where the value of the variable \(\ivar\) is updated to the value of the term \(\term\).
The \emph{test program} \(\ptest{\fml}\) proceeds without effect as long as \(\fml\) is true in the current state and aborts otherwise.
The \emph{non-deterministic choice program} \(\pchoice{\prog}{\progb}\) is a non-deterministic (existential) branch: the evolution may continue according to one of the two hybrid programs \(\prog\) or \(\progb\).
A \emph{sequential composition} \(\prog;\progb\) of hybrid programs executes \(\progb\) after \(\prog\).
For finite iteration, i.e., for \(\prog;\ldots;\prog\) repeated \(n\) times, write \(\prog^n\) and let \(\piter{\prog}{n}\equiv\ptest{\mathrm{true}\cup\prog^1\cup\prog^2\cup\ldots\cup\prog^n}\) stand for bounded non-deterministic repetition.
The \emph{iteration program} \(\prepeat{\prog}\) represents finite repetition of \(\prog\) for an arbitrary (non-deterministic) number of times.
Finally, the \emph{continuous program} \(\pevolvein{\synodevec{\ivar}{\term}}{\fml}\) describes the continuous evolution of the state along the solution of the differential equation for a non-deterministic amount of time, as long as \(\fml\), called the \emph{evolution domain}, is satisfied.

\subsubsection{Notation and Abbreviations.}
As usual, \(\geq\), \(=\), \(\limply\) and the dual connectives are defined via negation.
Moreover, the box modalities \(\dbox{\prog}\fml\) abbreviate \(\lnot\ddiamond{\prog}\lnot\fml\) to say that every execution of \(\prog\) ends in a state in which \(\fml\) is true.
Syntactically, for a tuple of terms \(\myvect{\term}\), the \(\infty\)-norm is used via the abbreviation
\begin{align*}
  \syntnorm{\myvect{\term}}<\termb & ~\equiv~ -\termb<\term<\termb\land\ldots\land -\termb<\term<\termb
\end{align*}
and similarly for \(\syntnorm{\myvect{\term}}\sim\termb\) where \({\sim}\in\{\leq,>,\geq\}\).

\subsubsection{Free Variables.}
The free variables of a formula are, as usual, the variables \(\ivar\) that appear somewhere in \(\fml\) not within the scope of a quantifier or assignment that \emph{necessarily} binds \(\ivar\).
For example, \(\ivar\) is free in \(\dbox{\humod{\ivarb}{\ivar}}\ivarb>0\) and \(\ivarb\) is not free.
(A formal definition is in the literature \cite[Definition 8]{DBLP:conf/cade/Platzer15}.)
The set of free variables of \(\fml\) is denoted \(\freevars{\fml}\).
A \emph{sentence} is a formula without free variables. 

\iflongversion
  \subsubsection{Polarity.}
  An instance of a subformula \(\fml\) inside a formula \(\fmlb\) is said to appear with \emph{positive polarity} if it is within the scope of an even number of negations and box modalities. (Since box modalities and tests can be nested.)
  Otherwise, it is said to appear with \emph{negative polarity}.
  To understand why box modalities affect polarity, consider, for example, the formula \(\dbox{\ptest{\fml}}\fmlb\).
  As this formula is equivalent to \(\lnot\fml\lor\fmlb\), it is important that the subformula \(\fml\) is considered to appear in negative polarity.
\else
\fi
 
\subsection{Semantics of Differential Dynamic Logic}\label{sec:dlsemantics}
The semantics of differential dynamic logic intuitively describes the intended behaviour of hybrid systems models.
First basic definitions and notations are introduced, before the denotational semantics of \dL \cite{DBLP:journals/jar/Platzer08} is recalled.
 
\subsubsection{Reachability Relations.}
Hybrid systems behaviour is semantically described by a \emph{reachability relation} \(\rel\) that describes which states can be reached from a given initial state.
The set \(\awinningregionp{\rel}{\openset}\) of states from which a target set \(\openset\subseteq\topspace\) of states is reachable along a relation \(\rel\subseteq \topspace\times\topspace\) is defined by
\(\awinningregionp{\rel}{\openset}=\{\topspaceelement :\mexists{\topspaceelementb\in\openset} (\topspaceelement,\topspaceelementb)\in\rel\}.\)
The composition of two relations \(\rel,\relb\) is
\(\rel\circ\relb=\{(\topspaceelement,\topspaceelementc) : \mexists{\topspaceelementb} (\topspaceelement,\topspaceelementb)\in\rel\mand(\topspaceelementb,\topspaceelementc)\in\relb\}\)
and \(\awinningregionp{(\rel\circ\relb)}{\openset}=\awinningregionp{\rel}{\awinningregionp{\relb}{\openset}}\).
The  \(n\)-fold composition of \(\rel\) with itself is written \(\rel^n=\rel\circ\ldots\circ\rel\).
The \emph{reflexive transitive closure} of a relation \(\rel\) is \(\transitiveclosure{\rel} = \bigcup_{n\geq 0}\rel^n\).

\subsubsection{Differential Equations.}
The continuous part of a hybrid system is usually described by an ordinary differential equation modeling the dynamics of motion.
In this paper, hybrid systems are specified symbolically, and attention is restricted to \emph{polynomial differential equations} \(\ode{x}{\vecfield}\) where \(\vecfield:\reals^n\to\reals^n\) is a polynomial with rational coefficients.

By the Picard-Lindelöf theorem \cite{walter}, rational polynomial differential equations locally have unique solutions for every initial value \(x_0\in\reals^n\) and every solution extends to a solution on an open maximal interval of existence.
Let \(I_{x_0}\subseteq \reals\) denote this maximal interval for the initial value problem \(\ode{x}{\vecfield}\), \(x(0)=x_0\) and let \(\odeflow{\vecfield}{x_0}{\cdot}:I_{x_0}\to\reals^n\) denote its unique solution.
The partial function \(\odeflowfunc{\vecfield}: \flowdomain{\vecfield}= \{(x_0,t)\in\reals^n\times\reals:t\in I_{x_0}\}\to\reals^n\) is the flow of the differential equation.
The flow is \emph{continuous} on the domain, and if the coefficients of the polynomial \(\vecfield_\alpha\) vary continuously in \(\alpha\), then \(\odeflowfunc{\vecfield_\alpha}(x_0,t)\) is continuous in \(\alpha\), \(x_0\) and \(t\).
Proofs of these facts can be found in the literature~\cite{walter}.

\subsubsection{Denotational Semantics}
The semantics of \dL is defined in the literature \cite{DBLP:conf/lics/Platzer12a} and recalled here for completeness.
A \emph{state} is a function \(\state:\variables\to\reals\) from the set of variables \(\variables\) to real numbers, which assigns a value \(\state(\ivar)\) to every variable \(\ivar\) and \(\states=\reals^\variables\) is the set of \emph{all} states equipped with the product topology.
For a tuple of variables \(\myvect{\ivar}\) and a tuple \(\myvect{\ival}\) of values of the same length, \(\staterepvarsbyvals[\state]{\myvect{\ivar}}{\myvect{\ival}}\) is the state that coincides with \(\state\) everywhere, except \(\staterepvarsbyvals[\state]{\myvect{\ivar}}{\myvect{\ival}}(\ivar_i)=\ival_i\) for \(\ivar_i\in\myvect{\ivar}\).
The restriction of a state \(\state\) to a tuple~\(\myvect{\ivar}\) of variables is the tuple \(\staterestr[\state]{\myvect{\ivar}}=(\state(\ivar_1),\ldots,\state(\ivar_n))\).

The semantics of a term with respect to a state \(\semterm[\state]{\term}\in\reals\) is a real number:
\begin{align*}
   & \semterm[\state]{\ivar}
  =
  \hspace{-1pt}
  \state(\ivar)
   &
   & \semterm[\state]{\ratconst}
  =
  \ratconst
   &
  \semterm[\state]{\term+\termb}
  =
  \hspace{-1pt}
  \semterm[\state]{\term}+\semterm[\state]{\termb}
   &
   &                             &
  \semterm[\state]{\term\cdot\termb}
  \hspace{-1pt}
  =
  \semterm[\state]{\term}\cdot\semterm[\state]{\termb}
\end{align*}
The semantics of a formula \(\fml\) is the set \(\sem{\fml}\subseteq\states\) of states which satisfy \(\fml\) and the semantics of a hybrid program \(\prog\) is the reachability relation \(\sem{\prog}\subseteq{\states}\times{\states}\) described by the program:
\begin{align*}
   & \sem{\term>\termb}
  =
  \{\state : \semterm[\state]{\term} > \semterm[\state]{\termb}\}
   &                    &
  \sem{\fml\lor\fmlb}
  =
  \sem{\fml}\union\sem{\fmlb}
  &                    &
  \sem{\lnot\fml}
  =
  \states\setminus\sem{\fml}
  \\
   &
  \sem{\lexists{\ivar}\fml}
  =
  \{\state : \mexists{\ival\in\reals} \staterepvarbyval[\state]{\ivar}{\ival}\in\sem{\fml}\}
   &
   &
  \sem{\ddiamond{\prog}\fml}
  =
  \semwrp{\prog}{\sem{\fml}}
  \\
   &
  \sem{\humod{\ivar}{\term}}
  =
  \{(\state, \staterepvarbyval[\state]{\ivar}{\semterm[\state]{\term}}) : \state\in \states\}
   &
   & \sem{\ptest{\fml}}
  =
  \{(\state, \state) : \state\in \sem{\fml}\}
   &                    &
  \sem{\prepeat{\prog}}
  =
  \transitiveclosure{\sem{\prog}}
  \\
   &
  \sem{\pchoice{\prog}{\progb}}
  =
  \sem{\prog}\union\sem{\progb}
   &                    &
  \sem{\prog;\progb}
  =
  \sem{\prog}\circ\sem{\progb}
\end{align*}
For a continuous program \(\synodevec{\ivar}{\term}\) with \(n\) variables and a state \(\state\), the polynomial vector field \(\syntvectorfield{\synodevec{\ivar}{\term}}{\state}:\reals^n\to\reals^n\) is defined by \(\syntvectorfield{\synodevec{\ivar}{\term}}{\state}(\myvect{\ival}) = \staterestr[{\semterm[{\staterepvarsbyvals[\state]{\myvect{\ivar}}{\myvect{\ival}}}]{\myvect{\term}}}]{\myvect{\ivar}},\)
and the semantics of a continuous program is
\begin{align*}
  \sem{\pevolvein{\synodevec{\ivar}{\term}}{\fml}}
  = &
  \{(\state,\stateevolve{\state}{\mtime}) : \mtime\geq 0\mand (\state,\mtime)\in \flowdomain{\syntvectorfield{\synodevec{\ivar}{\term}}{\state}}\mand\stateevolve{\state}{[0,\mtime]}\subseteq\sem{\fml}\}
  \\
    & \text{where } \stateevolve{\state}{\mtime}=\staterepvarsbyvals[\state]{\myvect{\ivar}}{\odeflow{\syntvectorfield{\synodevec{\ivar}{\term}}{\state}}{\staterestr[\state]{\myvect{\ivar}}}{\mtime}}
  \text{ is the state-flow}.
\end{align*}
As usual, a formula \(\fml\) is said to be \emph{valid} if \(\sem{\fml}=\states\).

\subsection{Robust Differential Dynamic Logic}
\label{sec:robustdl}

Robust differential dynamic logic (\rdL) is a careful refinement of \dL for the purpose of \emph{syntactically} ensuring that any formula describes properties with locally robust truth.
An example of such a formula is \(\term>\termb\).
If it is true in some interpretation of the variables, then it remains true under a sufficiently small perturbation of the variables.
However, the formula \(-\ivar^2\geq\ivar^2\) is \emph{not} robust, as it is true only in the \emph{razor-edge case} when \(
\ivar\) is interpreted as \(0\), and it is false for any other interpretation of \(\ivar\), no matter how close.

Robust \dL generalizes the distinction between strict inequality (\(>\)) and weak inequality (\(\geq\)) from real numbers to properties of hybrid systems generally.
There are two kinds of \rdL formulas: \emph{\strictfml formulas} \(\posfml\) generalizing strict inequality and \emph{\laxfml formulas} \(\negfml\) generalizing weak inequality, which are described by:
\begin{align*}
  \posfml,\posfmlb & ::= \term >\termb \mid \posfml\lor\posfmlb \mid \posfml\land\posfmlb \mid \lexists{\ivar}\posfml \mid \lforall{\ivar} \posfml \mid \ddiamond{\posprog} \posfml \mid \dbox{\negprog} \posfml
  &
    \text{\textcolor{annotgrey}{\strictfml \rdL}}
    \\
    \negfml,\negfmlb & ::= \term \geq \termb \mid \negfml\lor\negfmlb \mid \negfml\land\negfmlb \mid \lexists{\ivar} \negfml  \mid \lforall{\ivar} \negfml \mid \ddiamond{\negprog} \negfml \mid \dbox{\posprog} \negfml
    &
      \text{\textcolor{annotgrey}{\laxfml \rdL}}
  \end{align*}
The box and diamond modalities of \rdL formulas are again described by two different classes of hybrid programs.
Corresponding to the difference between \strictfml and \laxfml formulas, there are \emph{\exactprg} hybrid programs \(\posprog\) and \emph{\approxprog} hybrid programs \(\negprog\) defined by
\begin{align*}
  \posprog,\posprogb & ::= \humod{\ivar}{\term} \mid \ptest{\posfml} \mid \pchoice{\posprog}{\posprogb} \mid \posprog;\posprogb \mid \prepeat{\posprog} \mid \pevolvein{\synodevec{\ivar}{\term}}{\posfml}
  &
  \text{\textcolor{annotgrey}{\strictfml \rdL program}}
  \\
  \negprog,\negprogb & ::= \humod{\ivar}{\term} \mid \ptest{\negfml} \mid \pchoice{\negprog}{\negprogb} \mid \negprog;\negprogb \mid \prepeat{\negprog} \mid \pevolvein{\synodevec{\ivar}{\term}}{\negfml}
  &
      \text{\textcolor{annotgrey}{\approxprog \rdL program}}
\end{align*}
The difference between \exactprg and \approxprog programs is that tests and evolution domains are \strictfml formulas in \exactprg programs and \laxfml formulas in \approxprog programs, respectively.
Note that in a \strictfml formula the programs appearing in a diamond modality need to be \exactprg programs, while the programs in a box modality need to be \approxprog programs.
Dually, the diamond programs in a \laxfml formula are \approxprog programs, while the box programs in a \laxfml formula are \exactprg programs.

By design, robust differential dynamic logic consists of two \emph{fragments} of differential dynamic logic, since \(\posfml\) and \(\negfml\) \emph{are} formulas of differential dynamic logic.
The circle annotations are merely for emphasis, indicating that a formula belongs to the \strictfml (open \(\strictsymb{}\)) or the \laxfml (closed \(\laxsymb{}\)) fragment.
In fact, the \strictfml (\laxfml) formulas of \rdL correspond exactly to the fragment of \dL in which strict (weak) inequalities appear only in positions of positive polarity and weak (strict) inequalities appear only in positions of negative polarity.
\iflongversion
\else
  (For a full definition of polarity in \dL, see \citelongversion{sec:apppolar}.)
\fi

From a modeling and specification perspective, therefore, little practical expressiveness is lost, since the distinction between \(\ivar\geq0\) and \(\ivar>0\) is not of practical significance.
For example, hybrid systems controllers with equality constraints \(\ptest{\ivar=\term}\) may not appear in \exactprg programs.
This restriction is reasonable, as a controller cannot \emph{exactly verify} that the value of the real variable~\(\ivar\) is~\(0\), which would require checking infinitely many decimal places of \(\ivar\) \cite{DBLP:series/txtcs/Weihrauch00}.
Instead, a controller model may, more realistically, test \(\ptest{(\term-0.1<\ivar<\term+0.1)}\).

\subsubsection{Negation of Robust Formulas.}
Although negation is not a syntactic primitive of \rdL, it bridges the strict and the non-strict fragments:
The negation \(\lnotp{\posfmlb}\) of a \strictfml formula is a \laxfml formula and conversely the negation \(\lnotp{\negfml}\) of a \laxfml formula is a \strictfml formula.
For instance, the implication \(\negfml\limply\posfmlb\) is a \strictfml formula \(\lnotp{\negfml}\lor\posfmlb\) when \(\negfml\) is a \laxfml formula and \(\posfmlb\) is a strict formula.
Dually, the implication \(\posfmlb\limply\negfml\) is a \laxfml formula \(\lnotp{\posfmlb}\lor\negfml\).

\subsubsection{Robust Negation.}

As noted, the negation of a strict robust formula is not a strict robust formula.
\emph{Robust negation} is a slightly stronger notion defined by
\begin{align*}
   & \robnotp{\term>\termb} \equiv \termb>\term
   &                                            &
  \robnotp{\term\geq\termb} \equiv \termb\geq\term
\end{align*}
and extended to formulas homomorphically.
The robust negation \(\robnot{\posfml}\) of a \strictfml (\laxfml) \rdL formula \(\posfml\) is a \strictfml (\laxfml) \rdL formula.

Robust negation of \rdL formulas is a strengthening of negation that does not satisfy the law of the excluded middle.
That is, the robust negation \(\robnot{\posfml}\) entails the standard negation \(\lnot{\posfml}\), but not conversely.
A strict \rdL formula may be viewed as having one of three truth values in a state: \(\fml\) can be \emph{robustly true}, if it is semantically satisfied, \emph{robustly false} if \(\robnot{\fml}\) is satisfied and \emph{undetermined} if neither is the case.
For example, the formula \(\ivar>0\) is robustly true in a state with \(\state(\ivar)=1\) and robustly false in a state with \(\state(\ivar)=-1\).
However, in a state with \(\state(\ivar)=0\), the truth value of \(\ivar>0\) is not \emph{robustly determined}.
A \rdL formula~\(\fml\) is said to be \emph{\wellposed} if \(\fml\lor\robnot{\fml}\) is valid.

\subsubsection{Bounded Quantification.}
Bounded quantifiers are important in the definition of the reachability fragment of \rdL.
For a \strictfml formula \(\posfml\) and a term~\(\term\), which does not mention \(\ivar\), the formula \(\lexists{\ivar{>}\term}\posfml\) stands for the \strictfml formula \(\lexists{\ivar}(\ivar>\term \land \posfml)\).
Dually, \(\lforall{{\ivar{\geq}\term}}\posfml\) abbreviates the \strictfml formula \(\lforall{\ivar}(\ivar\geq\term\limply\posfml)\).
To exclude ambiguous notation such as \(\lexists{\ivar{>}\ivar}\posfml\), it is important that the term \(\term\) may not mention \(\ivar\).
For a closed interval \(\syninterv=[\term,\termb]\) of terms \(\term,\termb\) not mentioning~\(\ivar\), the \strictfml formula \(\lforallbdd[\syninterv]{\ivar}\posfml\) and the \laxfml formula \(\lexistsbdd[\syninterv]{\ivar}\negfml\) are defined by
\[
  \lforallbdd[\syninterv]{\ivar}\posfml\equiv\lforall{\ivar}((\term\leq\ivar\land \ivar\leq\termb)\limply\posfml)
  \qquad
  \lexistsbdd[\syninterv]{\ivar}\negfml\equiv\lexists{\ivar}(\term\leq\ivar\land\ivar\leq\termb \land\negfml)
\]
The \strictfml formula \(\lexistsbdd[\syninterv]{\ivar}\posfml\) and the \laxfml formula \(\lforallbdd[\syninterv]{\ivar}\negfml\) are defined analogously for the open interval \(\syninterv=(\term,\termb)\).

\subsection{Robust Reachability Fragment}\label{sec:robustreachdl}

Hybrid systems reachability properties are expressed using \emph{diamond modalities}.
Robust reachability differential dynamic logic \rrdL is a fragment of \rdL that contains arbitrary robust diamond modalities.
This makes it expressively powerful and at the same time tame enough to be complete.

\subsubsection{\Strictfml Robust Reachability.}
Syntactically, the \emph{\strictfml robust reachability} formulas and programs of robust reachability differential dynamic logic \rrdL are:
\begin{align*}
  \posrfml,\posrfmlb  & ::= \term >\termb \mid \posrfml\lor\posrfmlb \mid \posrfml\land\posrfmlb \mid \lexists{\ivar}\posrfml \mid \lforallbdd[\highlight{\syninterv}]{\ivar} \posrfml \mid \ddiamond{\posrprog} \posrfml \mid \dbox{\negrprog} \posrfml
  &
  \text{\textcolor{annotgrey}{\strictfml \rrdL}}
  \\
  \posrprog,\posrprogb & ::= \humod{\ivar}{\term} \mid \ptest{\posrfml} \mid \pchoice{\posrprog}{\posrprogb} \mid \posrprog;\posrprogb \mid \prepeat{\posrprog} \mid \pevolvein{\synodevec{\ivar}{\term}}{\posrfml}
  &
  \text{\textcolor{annotgrey}{\strictfml \rrdL}}
\end{align*}
where \(\syninterv=[\term,\termb]\) is an interval of terms, which do not mention \(\ivar\), and \(\negrprog\) is a \laxfml robust reachability program (see below).
Importantly, \strictfml robust reachability formulas capture essentially \emph{all} physically relevant reachability (diamond) properties, including unbounded continuous evolution and loop iteration.
The restriction that universal quantifiers are bounded means only that a \rrdL formula cannot \emph{assume} a variable to have \emph{any} arbitrary value.
For specification, it should always be possible to find an appropriate bounding range.
Moreover, \strictfml \rrdL does not restrict the reachability properties, which may feature \emph{unbounded} loops and unbounded continuous evolution.

\subsubsection{\Laxfml Robust Reachability.}
Although the focus of \rrdL is on reachability formulas, which need only \emph{diamond} modalities, the syntax of \strictfml \rrdL formulas admits a restricted set of safety formulas (box modalities) involving \laxfml formulas and programs of \rrdL.
\Laxfml \rrdL formulas do not allow iterations and all continuous evolutions need to be bounded (in space and time).
\begin{align*}
  \negrfml,\negrfmlb   & ::= \term \geq \termb \mid \negrfml\lor\negrfmlb \mid \negrfml\land\negrfmlb \mid \lexistsbdd[\highlight{\syninterv}]{\ivar}  \negrfml  \mid \lforall{\ivar} \negrfml \mid \ddiamond{\negrprog} \negrfml \mid \dbox{\posrprog} \negrfml
  &
  \text{\textcolor{annotgrey}{\laxfml \rrdL}}
  \\
  \negrprog,\negrprogb & ::= \humod{\ivar}{\term} \mid \ptest{\negrfml} \mid \pchoice{\negrprog}{\negrprogb} \mid \negrprog;\negrprogb \mid \pevolvein{\synodevec{\ivar}{\term}}{\negrfml\land \highlight{\syntnorm{\myvect{\ivar}}\leq \evolutionboundterm}}
  &
  \text{\textcolor{annotgrey}{\laxfml \rrdL}}
\end{align*}
The negation of a \strictfml \rrdL formula is a \laxfml \rrdL formula and vice versa.
This can be seen inductively as, for example, the negation \(\ddiamond{\negrprog}\lnot\posrfml\) of the \strictfml formula \(\dbox{\negrprog}\posrfml\) is a \laxfml formula.
The annotations~\(\rstrictsymb\) and~\(\scalebox{1.7}{\rlaxsymb}\) indicate that a formula belongs to one of the \emph{reachability} fragments consisting of \strictfml and \laxfml formulas and programs of \rrdL, respectively.

The definition of \laxfml \rrdL programs ensures that all box modalities in \strictfml \rrdL formulas are iteration-free and have bounded continuous evolution.
To see why this is necessary, observe that loops in box-modalities can encode universal quantifiers over the natural numbers \(\lforall{\ivar\in\naturals}\fml\) via \(\dbox{\humod{\ivar}{0};\prepeat{(\humod{\ivar}{\ivar+1})}}\fml\).
By Matiyasevich's theorem, the validity of such formulas cannot be semi-decidable 
and consequently there can be no complete proof calculus for \strictfml \rrdL formulas, if loops in box modalities are allowed.
For this reason, the restriction to the \strictfml \emph{reachability fragment}, where iteration programs may only appear in diamond formulas, is critical.

The assumption that the evolution of continuous programs in a box modality is bounded \(\syntnorm{\myvect{\ivar}}\leq \evolutionboundterm\) mirrors the boundedness restriction on universal quantifiers.
Both ensure that universal properties are restricted to a \emph{pre-defined range}.
The evolution bound also enforces an upper bound on the duration of the evolution, since the differential equation is assumed to include a time variable \(\tau'=1\).
All the restrictions involved in the robust reachability fragment are, to some extent, necessary to ensure that the fragment is natural and a complete calculus exists.

\subsection{Robustness Semantically}\label{sec:robustnesssem}

The crucial semantic property of \rrdL is that the semantics topologically reflects the robustness property which was enforced syntactically.
Assumptions should be \emph{topologically closed} and conclusions should be \emph{open}.
This is subtle, as it relies on continuity of flows of differential equations and needs to deal \emph{topologically} with infinite unions (existential quantifiers) and infinite intersections (universal quantifiers).
The key is a compactness argument on the \emph{domain of universal quantification}, which relies on a standard topological lemma.
\begin{lemmaE}[Tube Lemma]\label{lem:tubelemma}
  For \(\topspace,\topspaceb\) topological spaces, \(\topspaceb\) compact and \(\openset\subseteq\topspace\times \topspaceb\) open in the product topology, the set \(\{\topspaceelement\in\topspace: \mforall{\topspaceelementb{\in}\topspaceb} (\topspaceelement,\topspaceelementb)\in\openset\}\) is open.
\end{lemmaE}

\begin{proofE}
  Fix some \(\topspaceelement\in\{\topspaceelement\in\topspace: \mforall{\topspaceelementb{\in}\topspaceb} (\topspaceelement,\topspaceelementb)\in\openset\}\).
  For every \(\topspaceelementb\in\topspaceb\) pick open neighborhoods \(\opensetb_{\topspaceelementb},\opensetc_{\topspaceelementb}\) such that \((\topspaceelement,\topspaceelementb)\in \opensetb_{\topspaceelementb}\times \opensetc_{\topspaceelementb}\subseteq \openset\).
  As \(\topspaceb\) is compact, there are \(\topspaceelementb_1,\ldots,\topspaceelementb_n\) such that \(\topspaceb\subseteq\opensetc_{\topspaceelementb_1}\union\ldots\union\opensetc_{\topspaceelementb_n}\).
  Let \(\opensetb=\opensetb_{\topspaceelementb_1}\intersection\ldots\intersection\opensetb_{\topspaceelementb_n}\) and note that \(\opensetb\) is open, \(\topspaceelement\in\opensetb\) and \((\topspaceelement,\topspaceelementb)\in\openset\) for all \(\topspaceelementb\in\topspaceb\).
\end{proofE}
By induction it can be shown that \strictfml reachability formulas are semantically open and \laxfml reachability formulas are closed.
\Cref{lem:tubelemma} ensures that bounded universal quantification preserves these properties.

\begin{theoremE}[Semantic Robustness]\label{thm:semantictopology}
  For \rrdL formulas and programs:
  \begin{enumerate}
    \item \(\semopen{\posrfml}\) is open for \strictfml formulas~\(\posrfml\)\label{it:strictopen}
    \item \(\semclosed{\negrfml}\) is closed for \laxfml formulas~\(\negrfml\)\label{it:laxclosed}
    \item \(\semexactwrp{\posrprog}{\openset}\) is open for \exactprg programs \(\posrprog\) and open sets \(\openset\)\label{it:exopen}
    \item \(\semapproxwrp{\negrprog}{\closedset}\) is closed for \approxprog programs \(\negrprog\) and closed sets \(\closedset\)\label{it:approxclosed}
  \end{enumerate}
\end{theoremE}

\begin{proofE}
  Proceed by induction on the formulas and programs of \rrdL.
  Many cases are immediate, and only the interesting ones are shown.

  \begin{caselist}
    \caseof{\eqref{it:strictopen}}
    For a formula \(\lexists{\ivar}\posrfml\) note that \(\semopen{\lexists{\ivar}\posrfml} = \bigcup_{\ival\in\reals}\{\state:\staterepvarbyval[\state]{\ivar}{\ival}\in\semopen{\posrfml}\}\) is open as a union of open sets.

    Most interesting is the case for a bounded universal quantifier \(\lforallbdd[\syninterv]{\ivar}{\posrfml}\) where \(\syninterv\) is a closed (syntactic) interval with terms \([\term,\termb]\).
    For every state \(\state\), define the continuous function \(\gamma(s)= s \semterm[\state]{\term}+ (1-s)\semterm[\state]{\termb}\), so that
    \[\semopen{\lforallbdd[\syninterv]{\ivar}\posrfml} = \{\state:\mforall{s{\in}[0,1]}\staterepvarbyval[\state]{\ivar}{\gamma(s)}\in\semopen{\posrfml}\}.\]
    This set is open by \Cref{lem:tubelemma}, since \(\{(\state,s) \in\states\times[0,1]: \staterepvarbyval[\state]{\ivar}{\gamma(s)}\in\semopen{\posrfml}\}\) is open with respect to the relative topology on \(\states\times[0,1]\).

    For diamond and box formulas, openness follows from the induction hypotheses \eqref{it:exopen} and \eqref{it:approxclosed}.

    \caseof{\eqref{it:laxclosed}} This is the dual of the previous case.

    \caseof{\eqref{it:exopen}}
    For \(\prepeat{\posrprog}\) note that \(\semexactwrp{\prepeat{\posrprog}}{\openset}=\bigcup_{n\geq 0}\semexactrepeatp{\posrprog}{\openset}{n}\) is open as a union of open sets.

    For a continuous program \(\pevolvein{\synodevec{\ivar}{\term}}{\posrfml}\), as in the definition of the semantics, let \(\syntvectorfield{\synodevec{\ivar}{\term}}{\state}:\reals^n\to\reals^n\) be defined by \(\syntvectorfield{\synodevec{\ivar}{\term}}{\state}(\myvect{\ival}) = \staterestr[{\semterm[{\staterepvarsbyvals[\state]{\myvect{\ivar}}{\myvect{\ival}}}]{\myvect{\term}}}]{\myvect{\ivar}}\) and set \(\stateevolve{\state}{\mtime}=\staterepvarsbyvals[\state]{\myvect{\ivar}}{\odeflow{\syntvectorfield{\synodevec{\ivar}{\term}}{\state}}{\staterestr[\state]{\myvect{\ivar}}}{\mtime}}\).
    For continuous programs define the set of states reaching \(\openset\) in time \(\mtime\) as
    \[R_t(\openset) = \{\state: \stateevolve{\state}{\mtime}\in\openset\}\intersection\{\state: \mforall{\mtimeb{\in}[0,\mtime]} \stateevolve{\state}{\mtimeb}\in\semopen{\posrfml}\},\]
    so that
    \(\semexactwrp{\pevolvein{\synodevec{\ivar}{\term}}{\posrfml}}{\openset} = \bigcup_{\mtime\geq 0}R_t(\openset)\).
    Hence, it suffices to show that every \(R_t(\openset)\) is open.
    As the flow \(\odeflowfunc{\syntvectorfield{\synodevec{\ivar}{\term}}{\state}}\) is continuous in \(\state\), so is \(\stateevolvefunc{\state}\), and hence \(\{\state: \stateevolve{\state}{\mtime}\in\openset\}\) and \(\{(\state,\mtime) : \stateevolve{\state}{\mtime}\in\semopen{\posrfml}\}\) are open, since \(\openset\) and \(\semopen{\posrfml}\) are open by assumption and induction hypothesis, respectively.
    Moreover, by \Cref{lem:tubelemma} the set \(\{\state: \mforall{\mtimeb{\in}[0,\mtime]} \odeflow{\ivar}{\state}{\mtimeb}\in\semopen{\posrfml}\}\) is open.

    \caseof{\eqref{it:approxclosed}} The only case of interest is for continuous programs \(\pevolvein{\synodevec{\ivar}{\term}}{\negrfml\land \syntnorm{\myvect{\ivar}}\leq \evolutionboundterm}\).
    Define the functions \(\syntvectorfield{\synodevec{\ivar}{\term}}{\state}\) and \(\stateevolvefunc{\state}\) as before.
    Since \(\evolutionboundterm\) does not mention any variables from \(\myvect{\ivar}\), the evolution \(\evolutionboundtermsem=\semterm[\state]{\evolutionboundterm}\) is constant, i.e., \(\semterm[\stateevolvefunc{\state}]{t}(\evolutionboundterm)=\evolutionboundtermsem\) for all~\(t\).
    Similar to the previous case, for a closed set \(\closedset\) define
    \[R_t(\closedset) = \{\state: \stateevolve{\state}{\mtime}\in\closedset\}\intersection\{\state: \mforall{\mtimeb\in[0,\mtime]} \stateevolve{\state}{\mtimeb}\in\semclosed{\negrfml}\mand \norm[\infty]{\stateevolve{\state}{\mtimeb}}\leq \evolutionboundtermsem\}\]
    and again every \(R_t(\closedset)\) is closed by continuity of \(\stateevolvefunc{\state}\).
    Since the differential equation contains \(\timevar'=1\), the potential evolutions are time-bounded by \(\evolutionboundtermsem\), so that  \(\semapproxwrp{\pevolvein{\synodevec{\ivar}{\term}}{\negrfml\land \syntnorm{\myvect{\ivar}}\leq \evolutionboundterm}}{\closedset} = \{\state: \mexists{t{\in} [0,\evolutionboundtermsem]} \state\in R_t(\closedset)\}\).
    By \Cref{lem:tubelemma} this set is closed. 
  \end{caselist}
\end{proofE}
The defining restriction of \rrdL formulas that loops may not appear in box modalities is crucial for \eqref{it:approxclosed}.
The proof of \eqref{it:approxclosed} uses only time-boundedness.
Instead of assuming \(\syntnorm{\myvect{\ivar}}\leq\evolutionboundterm\) in the evolution domain, it suffices to assume \(\timevar\leq\evolutionboundterm\).
The proof goes through when allowing \approxprog programs \(\pevolvein{\synodevec{\ivar}{\term}}{\negrfml\land \timevar\leq \evolutionboundterm}\).

\section{Axioms and Rules of \dL}
\label{sec:calculus}

The power of differential dynamic logic comes from its proof calculus, which is relatively complete~\cite{DBLP:conf/lics/Platzer12a} and can prove all invariants and all algebraic properties of algebraic hybrid systems~\cite{DBLP:journals/jacm/PlatzerT20}.
Moreover, all true open properties of compact initial value problems are provable~\cite{DBLP:journals/jacm/PlatzerQ25}.
In \Cref{sec:completeness}, completeness for general robust reachability properties will be shown.
The proof calculus for \dL can be found in the literature~\cite{DBLP:conf/lics/Platzer12a,DBLP:journals/jacm/PlatzerT20}.
The relevant (derived) rules and axioms for the completeness proof are recalled.
Soundness for robust \dL is inherited from soundness for \dL, since \rdL is a syntactic fragment and their semantics coincide.

\subsection{Basic \dL Calculus}

The \dL calculus is a generalization of classical first-order sequent calculus and the calculus of first-order dynamic logic.
A sequent \(\lsequent{\fmlset}{\fmlsetb}\) consists of two (unordered) sets \(\fmlset,\fmlsetb\) of formulas.
As usual, the left-hand side of the sequent is read conjunctively and the right-hand side disjunctively.
The usual sequent comma notation, as in \(\fml,\fmlset_1,\fmlset_2=\{\fml\}\cup\fmlset_1\cup\fmlset_2\), is used.

\subsubsection{Real Arithmetic Foundation.}
The foundational level of differential dynamic logic is the theory of real arithmetic.
Because it is decidable \cite{Tarski10.1007/978-3-7091-9459-1_3}, real arithmetic is simply taken as the basic rule on which the proof calculus is built.
The rule

\begin{calculus}
  \cinferenceRule[qear|\usebox{\Rval}]{real arithmetic proof rule}
  {
    \linferenceRule[sequent]
    {\lclose}
    {\lsequent{\fmlsetb} {\fmlset}}\quad
  }{$\text{if}~\textstyle\landfold_{\fmlb\in\fmlsetb}\fmlb \limply \lorfold_{\fml\in\fmlsetb}\fml~\text{is valid in \reals}$}
\end{calculus}

\noindent
discharges any sequent in a proof tree that is semantically valid in \(\reals\), i.e., a tautology of real arithmetic.
The interesting and difficult parts of hybrid systems verification, therefore, arise from hybrid programs, in particular \emph{unbounded iteration} and \emph{continuous evolution}.
The remaining real arithmetic problems can be handled algorithmically by real arithmetic SMT solving \cite{cad10.1007/3-540-07407-4_17}.

\subsubsection{Basic Logical Proof Rules.}
The basic logical rules for the standard connectives are the usual sequent calculus rules.
Only some of the rules are needed for the completeness proof, and these are recalled here in adapted (derivable) form:

\begin{calculuscollection}
  \vspace{0.5em}
  \begin{calculus}
    \cinferenceRule[orR|$\lor$R]{or right proof rule}
    {
      \linferenceRule[sequent]
      {\lsequent{\fmlsetb} {\fml,\fmlb,\fmlset}}
      {\lsequent{\fmlsetb} {\fml\lor \fmlb,\fmlset}}
    }{}
    \cinferenceRule[existsR|$\exists$R]{exists right rule}
    {
      \linferenceRule[sequent]
      {\lsequent{\fmlsetb} {\ddiamond{\humod{\ivar}{\term}}\fml,\fmlset}}
      {\lsequent{\fmlsetb} {\lexists{\ivar}\fml,\fmlset}}\quad
    }{}
    \cinferenceRule[C|C]{equivalence rule}
    {
      \linferenceRule[sequent]
      {\lsequent{} {\fml\lbisubjunct\fmlb}
        &\lsequent{\fmlsetb} {\fmlb,\fmlset}}
      {\lsequent{\fmlsetb} {\fml,\fmlset}}
    }{}
  \end{calculus}
  \qquad\quad
  \begin{calculus}
    \cinferenceRule[andR|$\land$R]{and right proof rule}
    {
      \linferenceRule[sequent]
      {\lsequent{\fmlsetb} {\fml,\fmlset}
        &\lsequent{\fmlsetb} {\fmlb,\fmlset}
      }
      {\lsequent{\fmlsetb} {\fml\land \fmlb,\fmlset}}
    }{}
    \cinferenceRule[forallR|$\forall$R]{for all right rule}
    {
      \linferenceRule[sequent]
      {\lsequent{\fmlsetb,\ivar\in \syninterv} {\fml,\fmlset}}
      {\lsequent{\fmlsetb} {{}\lforallbdd[\syninterv]{\ivar}\fml,\fmlset}}\quad
    }{$\ivar\notin\freevars{\fmlsetb,\fmlset}$}
  \end{calculus}
  \vspace{0.5em}
\end{calculuscollection}

\noindent
The rules \irref{orR}, \irref{andR} and \irref{forallR} are standard.
The existence axiom \irref{existsR} is usually formulated in terms of substitution of \(\term\) for \(\ivar\) in~\(\fml\).
The formulation in terms of the assignment modality is used only for the convenience of concentrating all subtleties related to substitution in the single case of assignment programs.
Recall that \(\ivar\) does not appear in the bounds of the interval \(\syninterv\) in axiom \ref{forallR}.
To apply \irref{forallR} in a context where \(\ivar\) is free in \(\fmlset\) or \(\fmlsetb\), it is always possible to \emph{uniformly rename all} occurrences of \(\ivar\) in \(\fml\) before applying \irref{forallR}.
Rule \irref{C} is a simplified form of the cut rule that is used to instantiate the hybrid systems axioms below.

\subsubsection{Assignment Programs and Substitutions.}
Assignment hybrid programs are usually axiomatized by an axiom of the form \(\ddiamond{\humod{\ivar}{\term}}\fml\lbisubjunct \fmlreplacevarby[\fml]{\ivar}{\term}\), where \(\fmlreplacevarby[\fml]{\ivar}{\term}\) denotes the substitution of \(\term\) for \(\ivar\) in \(\fml\).
However, this causes unnecessary difficulties with free variable capture and inadmissible substitutions.
When \(\ivar\) is both free and bound, as for example in a repetition program \(\prepeat{(\humod{\ivar}{\ivar+1})}\), substitution is subtle.
To simplify this, an alternative axiomatization based on the equivalence
\[\ddiamond{\humod{\ivar}{\term}}\fml\lbisubjunct\lforall{\ivarc}(\ivarc=\ivar\limply\lforall{\ivar} (\ivar=\termreplacevarby[\term]{\ivar}{\ivarc}\limply \fml))\]
is used, where \(\ivarc\) is suitably fresh and \(\termreplacevarby[\term]{\ivar}{\ivarc}\) is the term \(\term\) with \(\ivar\) replaced uniformly by \(\ivarc\).
This avoids substitution into formulas by universally quantifying over \(\ivar\) and ensuring that it takes the correct value.
In case \(\term\) mentions \(\ivar\), the value of~\(\ivar\) is retained in the scope of the universal quantifier \(\lforall{\ivar}\) by means of the fresh variable \(\ivarc\).
For convenience, write \(\fmlsafereplacevarby[\fml]{\ivar}{\term}\equiv\lforall{\ivarc}(\ivarc=\ivar\limply\lforall{\ivar} (\ivar=\termreplacevarby[\term]{\ivar}{\ivarc}\limply \fml))\).

\subsubsection{Hybrid Program Axioms.}
In order to handle hybrid programs, \dL axiomatizes the hybrid program connectives exactly as dynamic logic:

\begin{calculuscollection}
  \vspace{0.5em}
  \begin{calculus}
    \cinferenceRule[diaassignequal|$\langle:=\rangle$]{safe diamond assignment axiom}
  {\ddiamond{\humod{\ivar}{\term}}\fml\lbisubjunct\fmlsafereplacevarby[\fml]{\ivar}{\term}\quad}
  {$\ivarc\notin\freevars{\term,\fml,\ivar}$}
  \cinferenceRule[diachoice|$\langle\cup\rangle$]{diamond choice axiom}
  {\ddiamond{\pchoice{\prog}{\progb}}\fml\lbisubjunct \ddiamond{\prog}\fml\lor\ddiamond{\progb}\fml}
  {}
  \cinferenceRule[diarepeat|$\langle*\rangle$]{diamond iteration axiom}
    {\ddiamond{\prepeat{\prog}}\fml\lbisubjunct \fml\lor\ddiamond{\prog}\ddiamond{\prepeat{\prog}}\fml}
    {}
\end{calculus}
\qquad
\begin{calculus}
  \cinferenceRule[diatest|$\langle?\rangle$]{diamond test axiom}
  {\ddiamond{\ptest{\fmlb}}\fml\lbisubjunct (\fmlb\land\fml)}
  {}
    \cinferenceRule[diacompose|$\langle{;}\rangle$]{box composition axiom}
    {\ddiamond{{\prog};{\progb}}\fml\lbisubjunct \ddiamond{\prog}\ddiamond{\progb}\fml}
    {}
    \cinferenceRule[starfinite|$\langle {}^{\leq n}\rangle$]{star finite rule}
    {
      \linferenceRule[sequent]
      {\lsequent{\fmlsetb} {\ddiamond{\piter{\prog}{n}}\fml,\fmlset}}
      {\lsequent{\fmlsetb} {\ddiamond{\prepeat{\prog}}\fml,\fmlset}}
    }{}
  \end{calculus}
\end{calculuscollection}

\noindent
Axiom \irref{diaassignequal} is the \emph{substitution-safe} version of the usual assignment axiom and axioms \irref{diatest}, \irref{diachoice} and \irref{diacompose} reduce program operations to logical connectives equivalently.
Axiom \irref{diarepeat} characterizes loop reachability and the derivable rule \irref{starfinite} says that in case \(\fml\) can be reached in at most \(n\)-steps, then an iteration of some length can reach \(\fml\).

\subsection{Continuous Program Axioms}
Continuous programs require a few additional axioms, such as \(\mathrm{dI}\) for differential invariance reasoning \cite{DBLP:journals/jacm/PlatzerT20}, differential cut \(\mathrm{dC}\) to accumulate information, \(\Delta\) for relative completeness \cite{DBLP:conf/lics/Platzer12a} and real induction axioms \cite{DBLP:journals/jacm/PlatzerT20} to globalize local arguments.
This section presents the axioms used in the \rrdL completeness proof.

\subsubsection{Continuous Safety-Reachability Duality.}
\label{sec:boxdiamondodeduality}
An important insight for the completeness proof is a robust duality between continuous reachability and continuous safety.
\emph{Robust} reachability and safety (box) properties of continuous programs are \emph{inter-translatable}.
Intuitively, an evolution is safe for bounded time and within a bounded domain if and only if it reaches the time bound or leaves the domain before leaving the safe region.

\begin{theoremE}[ODE Duality Soundness]
  The ODE duality axiom is sound
  \[
    \cinferenceRule[odeboxdiamond|${[}$ODE$\rangle$]{ODE box diamond axiom}
    { 
      \linferenceRule[viuqe]
      {
        \dbox{\pevolvein{\synodevec{\ivar}{\term}}{\negrfmlb\land \syntnorm{\myvect{\ivar}}\leq \evolutionboundterm}}\posrfml
      }{
        \ddiamond{\pevolvein{\synodevec{\ivar}{\term}}{\posrfml\lor \syntnorm{\myvect{\ivar}}> \evolutionboundterm}}(\lnot\negrfmlb\lor \syntnorm{\myvect{\ivar}}> \evolutionboundterm)}
    }{}
  \]
  for a \strictfml \rrdL formula \(\posrfml\), a \laxfml \rrdL formula \(\negrfmlb\) and a term \(\evolutionboundterm\) not mentioning variables from \(\myvect{\ivar}\).
\end{theoremE}

\begin{proofE}
  As in the definition of the semantics for continuous programs let \(\syntvectorfield{\synodevec{\ivar}{\term}}{\state}:\reals^n\to\reals^n\) be the vector field defined by \(\syntvectorfield{\synodevec{\ivar}{\term}}{\state}(\myvect{\ival}) = \staterestr[{\semterm[{\staterepvarsbyvals[\state]{\myvect{\ivar}}{\myvect{\ival}}}]{\myvect{\term}}}]{\myvect{\ivar}}\) and define the state-flow map \(\stateevolve{\state}{\mtime}=\staterepvarsbyvals[\state]{\myvect{\ivar}}{\odeflow{\syntvectorfield{\synodevec{\ivar}{\term}}{\state}}{\staterestr[\state]{\myvect{\ivar}}}{\mtime}}\).
  Now consider a state \(\state\in \sem{\lforall{\myvect{\ivar}}(\negrfmlb\limply\syntnorm{\myvect{\ivar}}\leq\evolutionboundterm)}\) and let \(\evolutionboundtermsem=\semterm[\state]{\evolutionboundterm}\).
  Clearly \(\semterm[\stateevolve{\syntvectorfield{\synodevec{\ivar}{\term}}{\state}}{\mtime}]{\evolutionboundterm}=\evolutionboundtermsem\) for all \(\mtime\) in the domain.

  For soundness of the \(\rightarrow\) implication, assume \(\state\in\sem{\dbox{\pevolvein{\synodevec{\ivar}{\term}}{\negrfmlb\land  \syntnorm{\myvect{\ivar}}\leq \evolutionboundterm}}\posrfml}\) and let
  \[\mtime=\inf\{\mtimeb\leq \evolutionboundtermsem: (\staterestr{\myvect{\ivar}},s)\notin\flowdomain{\syntvectorfield{\synodevec{\ivar}{\term}}{\state}}\text{ or }\stateevolve{\state}{\mtimeb} \notin\semclosed{\negrfmlb\land \syntnorm{\myvect{\ivar}}\leq \evolutionboundterm}\}.\]
  There are two cases: either \(\mtime=\infty\) or \(\mtime\leq\evolutionboundtermsem\).

  In case \(\mtime=\infty\), note that
  \((\state,\evolutionboundtermsem)\in \flowdomain{\syntvectorfield{\synodevec{\ivar}{\term}}{\state}}\) and \(\stateevolve{\state}{[0,\evolutionboundtermsem]}\in\semclosed{\negrfmlb\land\syntnorm{\myvect{\ivar}}\leq\evolutionboundterm}\).
  Thus, by assumption \(\stateevolve{\state}{[0,\evolutionboundtermsem]}\in\semopen{\posrfml}\).
  Since the interval of solution existence is open, there is some \(\mtimeb>\evolutionboundtermsem\) such that \((\state,\mtimeb)\in \flowdomain{\syntvectorfield{\synodevec{\ivar}{\term}}{\state}}\).
  Since the differential equation includes \(\timevar'=1\) then also \(\stateevolve{\state}{[0,\mtimeb]}\subseteq {\semopen{\posrfml\lor\syntnorm{\myvect{\ivar}}>\evolutionboundterm}}\) and \(\stateevolve{\state}{\mtimeb}\in \semopen{\syntnorm{\myvect{\ivar}}>\evolutionboundterm}\).
  Hence \(\state\in\semopen{\ddiamond{\pevolvein{\synodevec{\ivar}{\term}}{\posrfml\lor \syntnorm{\myvect{\ivar}}>\evolutionboundterm}}(\lnot\negrfmlb\lor \syntnorm{\myvect{\ivar}}>\evolutionboundterm)}\).

  In case \(\mtime\leq \evolutionboundtermsem\) note that \(\stateevolve{\state}{[0,\mtime)}\in\semclosed{\negrfmlb\land\syntnorm{\myvect{\ivar}}\leq\evolutionboundterm}\). 
  Note that \(\mtime\) is in the existence interval, as otherwise the solution would have escaped \(\syntnorm{\myvect{\ivar}}\leq\evolutionboundterm\) before then.
  By continuity of \(\stateevolvefunc{\state}\) as \(\semclosed{\negrfml\land\syntnorm{\myvect{\ivar}}\leq\evolutionboundterm}\) is closed by \Cref{thm:semantictopology} also \(\stateevolve{\state}{[0,\mtime]}\in\semclosed{\negrfmlb\land\syntnorm{\myvect{\ivar}}\leq\evolutionboundterm}\).
  So by the assumption \(\stateevolve{\state}{[0,\mtime]}\in\semopen{\posrfml}\).
  By continuity of the flow and openness of \(\semopen{\posrfml}\) (\Cref{thm:semantictopology}), there is some \(\mtimeb>\mtime\) such that \(\stateevolve{\stateb}{[0,\mtimeb]}\in\semopen{\posrfml}\).
  By definition of \(\mtime\) without loss of generality \(\stateevolve{\state}{\mtimeb}\notin\semclosed{\negrfmlb\land\syntnorm{\myvect{\ivar}}\leq\evolutionboundterm}\).
  Thus \(\state\in\semopen{\ddiamond{\pevolvein{\synodevec{\ivar}{\term}}{\posrfml\lor \syntnorm{\myvect{\ivar}}> \evolutionboundterm}}(\lnot\negrfmlb\lor \syntnorm{\myvect{\ivar}}> \evolutionboundterm)}\).

  For soundness of the \(\leftarrow\) implication, assume \[\state\in\semopen{\ddiamond{\pevolvein{\synodevec{\ivar}{\term}}{\posrfml\lor \syntnorm{\myvect{\ivar}}> \evolutionboundterm}}(\lnot\negrfmlb\lor \syntnorm{\myvect{\ivar}}> \evolutionboundterm )}\]
  Then there is \(\mtimeb\) such that \(\stateevolve{\state}{[0,\min\{\evolutionboundtermsem,\mtimeb\}]}\subseteq \semopen{\posrfml}\) and either \(\norm{\stateevolve{\state}{\mtimeb}}>\evolutionboundtermsem\) or \(\stateevolve{\state}{\mtimeb}\notin\semclosed{\negrfmlb}\).
  It follows immediately that \(\state\in\semopen{\dbox{\pevolvein{\synodevec{\ivar}{\term}}{\negrfmlb\land \syntnorm{\myvect{\ivar}}\leq \evolutionboundterm}}\posrfml}\).
\end{proofE}
The duality axiom is particular to \rrdL.
As it relies on robustness, it does not hold for \dL formulas, as the following incorrect instance with the \strictfml evolution domain formula \(\fmlb\equiv\ivar<0\) demonstrates:
\[\dbox{\pevolvein{\synodevec{\ivar}{\term}}{x<0\land \syntnorm{\myvect{\ivar}}\leq 1}}x<0\lbisubjunct
  \ddiamond{\pevolvein{\synodevec{\ivar}{\term}}{x<0\lor \syntnorm{\myvect{\ivar}}> 1}}(x\geq 0\lor \syntnorm{\myvect{\ivar}}> 1)
\]
The topological properties of the semantics of \rrdL formulas prevent such gaps.

Soundness of the axiom \irref{odeboxdiamond} crucially relies on the space-boundedness and the time-boundedness of the differential equation.
Time-boundedness is ensured by \(\syntnorm{\myvect{\ivar}}\leq \evolutionboundterm\) as a time component \(\tau'=1\) is always included in \(\pevolve{\synodevec{\ivar}{\term}}\).
The space-boundedness is important, as it makes it possible to handle solutions which exist only for finite time and therefore leave the bounded domain.

\subsubsection{Differential Equations via Iterations.}

\newcommand{\iterationprog}{\ptest{(\syntnorm{\myvect{\ivar}}< \evolutionboundterm-\varepsilon)};\humod{\myvect{\ivar}}{\myvect{\ivar}+h\myvect{\term}};\humod{\varepsilon}{(1+h\lipsch)\varepsilon+\tfrac{\lipsch\fbdd}{2}h^2}}

An important ingredient to handle the continuous parts of hybrid systems is a reduction of the continuous dynamics to discrete dynamics by means of an Euler approximation similar to axiom~\sgandalfsymb \cite{DBLP:conf/lics/Platzer12a}.
This is possible because robustness ensures that if a state satisfying \(\fml\) is reachable, then there is some margin of error \(\varepsilon>0\) and a reachable state \(\state\), such that \(\fml\) is true everywhere in the \(\varepsilon\)-neighborhood for \(\state\).
Thus, continuous reachability can be rephrased by reachability along a discrete Euler approximation of error at most \(\varepsilon\).
This \emph{quantitative robustness} is syntactically described by the \emph{\(\varepsilon\)-interior} \(\epsint[\varepsilon]{\fml}\) of a \strictfml formula \(\fml\):
\[\epsint[\varepsilon]{\fml}\equiv
  \lforallbdd[{[\ivar_1-\varepsilon,\ivar_1+\varepsilon]}]{\ivarb_1} \ldots   \lforallbdd[{[\ivar_n-\varepsilon,\ivar_n+\varepsilon]}]{\ivarb_n} \fmlreplacevarby[\fml]{\myvect{\ivar}}{\myvect{\ivarb}}
\]
which requires \(\fml\) to be true in a neighborhood of \(\myvect{\ivar}\).

To symbolically bound the error of the Euler approximation of a differential equation \(x'=f(x)\) on a compact set, bounds on \(f\) and the Jacobian \(Df\), respectively, are required.
For a syntactic differential equation \(\synodevec{\ivar}{\term}\)
let \(\tfrac{\partial\term_i}{\partial\ivar_j}\) denote the formal partial derivative of the polynomial term \(\term_i\) with respect to variable \(\ivar_j\).
The formula
\(\secondderbound\equiv\lforall{\syntnorm{\myvect{\ivar}}{\leq} \evolutionboundterm+2} (\syntnorm{\term}<\fbdd\land\textstyle\bigwedge_{i,j}\syntnorm{\tfrac{\partial\term_i}{\partial\ivar_j}}<\lipsch)\)
ensures that the syntactic vector field is bounded by \(\fbdd\) and its Jacobian is bounded by \(\lipsch\) in the \(\infty\)-norm on the set defined by \(\syntnorm{\myvect{x}}\leq \evolutionboundterm+2\).

The Euler-step program \(\iterationprogabb\equiv\iterationprog\) describes a single iteration step of an Euler approximation with step size \(h\) and accumulates the standard local error bound in the variable \(\varepsilon\).
The program ensures that the approximation remains in the compact domain \(\syntnorm{\myvect{\ivar}}<  \evolutionboundterm-\varepsilon\).

\begin{theoremE}[Euler Axiom Soundness]
  The Euler axiom \irref{diaode} is sound

  \begin{calculuscollection}
    \begin{calculus}
      \cinferenceRule[diaode|$\langle{E}\rangle$]{diamond ode axiom}
      {\linferenceRule[equivl]
        {
          \lexists{\evolutionboundterm,\fbdd,\lipsch{>}0}(\secondderbound\land\lexists{h{>}0}\ddiamond{\humod{\varepsilon}{0};\prepeat{(\ptest{\epsint[\varepsilon]{\posrfmlb}};\iterationprogabb)}}(\epsint[\varepsilon]{\posrfml}\land \varepsilon <1))
        }
        {
          \ddiamond{\pevolvein{\synodevec{\ivar}{\term}}{\posrfmlb}}\posrfml
        }
      }
      {}
    \end{calculus}\\
  \end{calculuscollection}
  where \(\posrfml \text{ is a \strictfml}\) \rrdL formula and \(h,\evolutionboundterm,\fbdd,\lipsch,\varepsilon\notin\freevars{\posrfml}\cup\{\myvect{\ivar},\myvect{\term}\}\) are fresh.
\end{theoremE}

\begin{proofE}
  By \irref{diaodebound} it suffices to show equivalently
  \begin{align*}
                      & {
        \ddiamond{\pevolvein{\synodevec{\ivar}{\term}}{\posrfmlb\land\syntnorm{\myvect{\ivar}}{<}\evolutionboundterm}}\posrfml
    }                     \\
    \lbisubjunct\quad &
    {
        \lexists{\fbdd,\lipsch{>}0}(\secondderbound\land\lexists{h{>}0}\ddiamond{\humod{\varepsilon}{0};\prepeat{(\ptest{\epsint[\varepsilon]{\posrfmlb}};\iterationprogabb)}}(\epsint[\varepsilon]{\posrfml}\land \varepsilon <1))
      }
  \end{align*}

  As in the definition of the semantics of continuous programs, let \(\syntvectorfield{\synodevec{\ivar}{\term}}{\state}:\reals^n\to\reals^n\) be \(\syntvectorfield{\synodevec{\ivar}{\term}}{\state}(\myvect{\ival}) = \staterestr[{\semterm[{\staterepvarsbyvals[\state]{\myvect{\ivar}}{\myvect{\ival}}}]{\myvect{\term}}}]{\myvect{\ivar}}\) and let \(\stateevolve{\state}{\mtime}=\staterepvarsbyvals[\state]{\myvect{\ivar}}{\odeflow{\syntvectorfield{\synodevec{\ivar}{\term}}{\state}}{\staterestr[\state]{\myvect{\ivar}}}{\mtime}}\) be the state-flow.
  Fix a state \(\state\) and note that \(\evolutionboundtermsem{\defeq}\semterm[\state]{\evolutionboundterm}\) is constant along the evolution.

  For the forward implication, suppose there is \(\mtime\) such that \(\stateevolve{\state}{\mtime}\in\semopen{\posrfml}\), \(\stateevolve{\state}{[0,t]}\subseteq\semopen{\posrfmlb}\) and \(\norm[\infty]{\stateevolve{\state}{\mtimeb}}<\evolutionboundtermsem\) for all \(s\in[0,t]\).
  Because \(\semopen{\posrfml}\) and \(\semopen{\posrfmlb}\) are open by \Cref{thm:semantictopology}, there is \(\delta<1\) such that
  \begin{align*}
      \{\staterepvarsbyvals[\state]{\myvect{\ivar}}{\myvect{\ival}} : \norm[\infty]{\myvect{\ival}-\odeflow{\syntvectorfield{\synodevec{\ivar}{\term}}{\state}}{\staterestr[\state]{\myvect{\ivar}}}{\mtime} }<\delta\}&\subseteq\semopen{\posrfml},
    \\
      \{\staterepvarsbyvals[\state]{\myvect{\ivar}}{\myvect{\ival}} : \mtimeb\in[0,\mtime]\mand\norm[\infty]{\myvect{\ival}-\odeflow{\syntvectorfield{\synodevec{\ivar}{\term}}{\state}}{\staterestr[\state]{\myvect{\ivar}}}{\mtimeb} }<\delta\}&\subseteq\semopen{\posrfmlb}.
  \end{align*}
  Then let
  \begin{equation}
    \mfbdd > \sup_{\norm[\infty]{x}\leq k+2}\norm[\infty]{\syntvectorfield{\synodevec{\ivar}{\term}}{\state}(x)}\quad\text{and}\quad\mlipsch> \sup_{\norm[\infty]{x}\leq k+2}\norm[\infty]{D\syntvectorfield{\synodevec{\ivar}{\term}}{\state}(x)},\label{eq:lipschbounds}
  \end{equation}
  so that \(\staterepvarsbyvals[\state]{\fbdd,\lipsch}{\mfbdd,\mlipsch}\in \semopen{\secondderbound}\).
  By the convergence of Euler's method, there is a step size \(H>0\) such that the sum of the local approximation errors (accumulated in \(\varepsilon\)) in time \(\mtime\) is bounded by \(\tfrac{\delta}{2}\).
  The Euler approximation with step size \(H\) then witnesses
  \[\staterepvarsbyvals[\state]{\fbdd,\lipsch,h}{\mfbdd,\mlipsch,H}\in \semopen{\ddiamond{\humod{\varepsilon}{0};\prepeat{(\ptest{\epsint[\varepsilon]{\posrfmlb}};\iterationprogabb)}}(\epsint[\varepsilon]{\posrfml}\land \varepsilon <1)}.\]

  For soundness of the backward implication, suppose there are \(\mfbdd\) and \(\mlipsch\) satisfying \eqref{eq:lipschbounds}.
  Moreover, assume that there is a step size \(h>0\) and a number of iterations \(n\), such that the \(n\)-step Euler iteration has global error bounded by some \(\delta<1\), the \(\delta\)-neighborhood of each Euler step remains in \(\semopen{\posrfmlb}\) and the \(\delta\)-neighborhood of the \(n\)-th Euler step is contained in \(\semopen{\posrfml}\).
  If the solution to the differential equation exists for time \(\mtime=nh\), then convergence of Euler's method ensures that \(\stateevolve{\state}{\mtime}\in\semopen{\posrfml}\), \(\stateevolve{\state}{[0,t]}\subseteq\semopen{\posrfmlb}\) and \(\norm[\infty]{\stateevolve{\state}{\mtimeb}}\leq\evolutionboundtermsem\).

  It remains to show that the solution exists for time \(\mtime=nh\).
  Assume for a contradiction that this is not the case.
  Then there must be \(m<n\) such that the flow leaves the \((\evolutionboundtermsem+1)\)-bounded region by time \(\mtimeb = mh<\mtime\), i.e.
  \(\norm[\infty]{\odeflow{\syntvectorfield{\synodevec{\ivar}{\term}}{\state}}{\staterestr[\state]{\myvect{\ivar}}}{\mtimeb}}>\evolutionboundtermsem+1\).
  Assume without loss of generality \(\norm[\infty]{\odeflow{\syntvectorfield{\synodevec{\ivar}{\term}}{\state}}{\staterestr[\state]{\myvect{\ivar}}}{\mtimeb'}}<\evolutionboundtermsem+2\) for all \(\mtimeb'\in[0,\mtimeb]\).
  As the bounds of \eqref{eq:lipschbounds} apply in the \(\evolutionboundtermsem+2\) bounded region, the first \(m\)-steps of the Euler iteration are within distance~\(\delta\) of the actual solution.
  This is a contradiction, as \(\delta<1\) and the Euler iteration steps remain in the region bounded by \(\evolutionboundtermsem\) by the definition of \(\iterationprogabb\).
\end{proofE}

Observe that in \irref{diaode} both sides of the equivalence are \strictfml formulas.
It is critical for soundness that the postcondition \(\posrfml\) is syntactically a \strictfml formula and, thus, semantically an \emph{open set}.
This allows some (arbitrarily small) wiggle room and sound, symbolic approximations become feasible, since the robust postcondition \(\posrfml\) can be equivalently replaced by a quantitatively robust strengthening \(\epsint[\varepsilon]{\posrfml}\).
For soundness of the forward direction, it suffices to show that the error accumulation correctly bounds the error.
For the converse, the evolution bound \(\evolutionboundterm\) plays a critical role in ensuring solution existence.
Within the bounds, the Euler approximation remains close to the real solution in the region \(\syntnorm{\myvect{\ivar}}<\evolutionboundterm+2\) and does not leave before reaching \(\epsint[\varepsilon]{\posrfml}\) by the iterative assumption.
Hence, the actual solution must remain bounded before reaching \(\epsint[\varepsilon]{\posrfml}\).
In particular, the solution exists at least until it reaches \(\posrfml\).
The Euler axiom \irref{diaode} is closely related to the differential equation characterization axiom \sgandalfsymb \cite{DBLP:conf/lics/Platzer12a}, which does not handle the possibility of divergent evolutions and postconditions with modalities.

To summarize, the proof calculus for \dL and, thus, for \rrdL is based on the basic rule \irref{qear}, which discharges valid formulas of real arithmetic.
The calculus extends \irref{qear} and basic propositional logical rules by axioms for (discrete) hybrid programs and axioms for continuous evolution including \irref{odeboxdiamond} and \irref{diaode}.

\section{Completeness for Robust Hybrid Systems Reachability}
\label{sec:completeness}

The main result of this paper is the proof of completeness of robust reachability differential dynamic logic.
Unlike previous completeness proofs, it does not rely on an undecidable oracle or apply only to algebraic hybrid systems \cite{DBLP:conf/lics/Platzer12a,DBLP:journals/jacm/PlatzerT20}.
Completeness for \rrdL is proved by an inductive argument, which step-by-step reduces the complexity of a formula to a basic formula of real arithmetic, which can be discharged by rule \irref{qear} and decided computationally \cite{cad10.1007/3-540-07407-4_17}.
There are two sources of difficulty.
First, the induction is transfinite, as it needs to handle all finite iterations \(\prog^n\) before handling loop programs \(\prepeat{\prog}\).
Second, the topological robustness properties need to be exploited to reduce all formulas to simpler valid ones.
The loops pose a particular challenge, as they cannot be reduced immediately to \emph{semantically equivalent} simpler formulas.

The proof of the following completeness theorem is by transfinite induction on a well-founded order defined in \Cref{sec:inductionorder}.
The proof is sketched in \Cref{sec:maincompletenesssect}, and the complete proof is in \citelongversion{sec:appendixproof}.
\begin{theorem}[Completeness for Robust Reachability Formulas]\label{thm:maincompleteness}
  Every valid strict robust reachability sentence \(\posrfml\) is provable: \(\valid{\posrfml}\) iff \(\prov{\posrfml}\).
\end{theorem}

\subsection{Induction Order}\label{sec:inductionorder}
The transfinite induction is based on an order on all \dL formulas, which is used to sequence the induction steps.
This order makes it possible to reduce all formulas to simpler ones in a way that is compatible with the proof calculus.
To define the order, every formula and every program of \dL is assigned an \emph{ordinal rank} \(\rankof{\cdot}\) by structural induction and the relations \(\provrelhalf\) and \(\provrel\) are defined by
\begin{align*}
  \fml\provrel\fmlb     & \iff \rankof{\fml} < \rankof{\fmlb}    \qquad &
  \fml\provrelhalf\fmlb & \iff \rankof{\fml} \leq \rankof{\fmlb}
\end{align*}
The full formal definition of the rank is in \citelongversion{sec:rank}.
The important properties of the order for the inductive proof are summarized in \Cref{prop:orderprop}.

\begin{proposition}\label{prop:orderprop}
  The relation \(\provrel\) on the set of \dL formulas is well-founded and
  \begin{enumerate}
    \item \(\fml,\fmlb\provrelhalf\fml\lor\fmlb\) and \(\fml,\fmlb \provrel\fml\land\fmlb\)\label{it:orderconju}\label{it:orderdisj}
    \item \(\ddiamond{\humod{\ivar}{\term}}{\fml}\provrel\lexists{\ivar}\fml\) and \(\fml\provrel\lforall{\ivar}\fml\) for all terms \(\term\)\label{it:existentialred}
    \item \(\fmlsafereplacevarby[\fml]{\ivar}{\term}\provrel\ddiamond{\humod{\ivar}{\term}}{\fml}\)\label{it:orderassign}
    \item \(\fmlb\land\fml\provrel\ddiamond{\ptest{\fmlb}}{\fml}\)\label{it:ordertest}
    \item \(\ddiamond{\prog}{\fml}\lor \ddiamond{\progb}{\fml}\provrel\ddiamond{\pchoice{\prog}{\progb}}{\fml}\)\label{it:orderchoose}
    \item \(\ddiamond{\prog}\ddiamond{\progb}{\fml}\provrel\ddiamond{{\prog};{\progb}}{\fml}\)\label{it:ordercompose}
    \item \(\ddiamond{\piter{\prog}{n}}{\fml}\provrel\ddiamond{\prepeat{\prog}}{\fml}\)\label{it:repeat}
    \item \(\ddiamond{\prepeat{(\ptest{\epsint[\varepsilon]{\fmlb}};\iterationprogabb)}}(\epsint[\varepsilon]{\fml}\land\varepsilon<1)\provrel \ddiamond{\pevolvein{\synodevec{\ivar}{\term}}{\fmlb\land \syntnorm{\myvect{\ivar}}<\evolutionboundterm}}\fml\) \label{it:orderode}
  \end{enumerate}
\end{proposition}

As the completeness proof proceeds by induction on sequents, the order is extended to sets of formulas.
Recall that both sides of a sequent \(\lsequent{\fmlsetb}{\fmlset}\) consist of a finite (unordered) sets of \dL formulas \(\fmlset,\fmlsetb\).
On the set \(\setoffmlsets\) of all sets of finite (unordered) sets of formulas define the \emph{induction order} \(\provsetrel\) and \(\provsetrelhalf\) as follows:
\begin{align*}
  \fmlset_1\provsetrel\fmlset_2     & \iff \max\{\rankof{\fml}:\fml\in\fmlset_1\} < \max\{\rankof{\fml}:\fml\in\fmlset_2\}    \\
  \fmlset_1\provsetrelhalf\fmlset_2 & \iff \max\{\rankof{\fml}:\fml\in\fmlset_1\} \leq \max\{\rankof{\fml}:\fml\in\fmlset_2\}
\end{align*}
The order \(\provsetrel\) is well-founded on \(\setoffmlsets\).
If \(\fmlset_1\provsetrel\fmlset_3\) and \(\fmlset_2\provsetrel\fmlset_3\), then \(\fmlset_1,\fmlset_2\provsetrel\fmlset_3\) and the order is invariant under variable substitutions.

\subsection{Proof of Completeness}\label{sec:maincompletenesssect}

Finally, the induction order can be used to prove completeness for \rrdL step-by-step.
This section provides some intuition and the full proof can be found in \citelongversion{sec:appendixproof}.
While the completeness result is stated only for sentences, the inductive proof shows more generally that every \emph{valid} sequent of the form
\[\lsequent{\ivar_1\in \syninterv_1,\ldots,\ivar_n\in\syninterv_n}{\fmlset}\quad \text{for } \freevars{\fmlset}\subseteq \{\ivar_1,\ldots,\ivar_n\} \text{, \(\fmlset\) \strictfml \rrdL formulas}\]
is provable.
The assumptions on the left-hand side bound all the free variables of \(\fmlset\).
The proof proceeds by induction on \(\fmlset\) with respect to the \(\provsetrel\) order.
In the inductive step, the \strictfml \rrdL formulas in \(\fmlset\) are shown to be derivable from valid sequents with \(\provsetrel\)-lower right-hand side.
The topological interpretation of robustness (\Cref{thm:semantictopology}) plays an important role.
For example, when reducing an existential quantifier in a valid sentence \(\lexists{\ivar}\posrfml\), there is a \emph{rational witness} \(\ratconst\), since \(\posrfml\) denotes an open set (\Cref{thm:semantictopology}).
As rationals are expressible syntactically, the rule \irref{existsR} can be used to instantiate the quantifier.
This retains validity and reduces the rank of the formula by \eqcref{prop:orderprop}{it:existentialred}.

The most interesting and crucial step is the case of loop programs \(\prepeat{\prog}\) in diamond modalities. (Loops in box modalities are not allowed in \rrdL.)
Since all free variables are bounded by the left-hand side of the sequent, the number of iterations needed can be bounded by a \emph{constant} using a compactness argument.
In this way, an unbounded loop can be reduced via \irref{starfinite} to a bounded loop, which is of lower complexity by \eqcref{prop:orderprop}{it:repeat}.

A continuous program in a diamond modality can be reduced to the case of an iteration program by axiom \irref{diaode}.
Since continuous programs in box modalities have bounded evolution domains, the case of such programs reduces to the case of continuous programs in diamond modalities by axiom \irref{odeboxdiamond}.

An interesting corollary to completeness is that the validity of \rrdL sentences and, thus, of robust \dL reachability problems of hybrid systems is semi-decidable:

\begin{corollary}
  Validity of strict \rrdL sentences \(\posrfml\) is semi-decidable.
\end{corollary}

In fact, the completeness proof gives a constructive algorithm for constructing a witnessing proof tree computably and proof search terminates.
If a strict \rrdL sentence \(\posfml\) is \wellposed and \emph{all} quantifiers are bounded, then one of the strict sentences \(\posfml\) or \(\robnot{\posfml}\) is valid:

\begin{corollary}
  Validity of \wellposed, bounded strict \rrdL sentences is decidable.
\end{corollary}

However, in general, it is difficult to know whether a given strict \rrdL sentence is \wellposed or not.

The \emph{doubly exponential} computational complexity of quantifier elimination algorithms \cite{DBLP:journals/jsc/DavenportH88}, which underlie the proof rule \irref{qear}, is one of the major bottlenecks for scalable symbolic hybrid systems verification.
However, the completeness (\Cref{thm:maincompleteness}) of \rrdL relies only on the rule \irref{qear} for the computationally more tractable case of \emph{quantifier-free \strictfml} formulas, which has \emph{singly exponential} complexity \cite{DBLP:journals/jsc/Renegar92}.

\section{Conclusion}
\label{sec:conclusion} 

The paper bridges the long-standing gap between approximate approaches (e.g., \(\delta\)-decidability) and exact deductive approaches to hybrid systems reachability verification.
By identifying robust differential dynamic logic as a syntactic fragment of \dL, the exact semantics is retained and \emph{absolute} completeness is gained.
Unlike in previous completeness proofs, \emph{no undecidable oracle} is required.
The central insight is that robustness is ensured by a simple syntactic restriction in \dL that enforces strict inequalities.
Crucially, robustness is purely \emph{topological} and does not involve quantitative noise tolerance and margins of error.
Consequently, working with robust \dL formulas entails no loss of convenience or practical expressiveness compared to standard \dL.

This work opens several areas of future research.
In particular, further investigation of the robust fragment for \emph{safety} questions is promising.
Moreover, an implementation of a proof tactic in the KeYmaera~X theorem prover for \dL can be an impactful tool to automatically verify robust reachability properties.

\begin{credits}
  \subsubsection{\ackname}
  This work has been supported by an Alexander von Humboldt Professorship.
  We thank the anonymous reviewers for their detailed comments, which helped to improve the paper.
\end{credits}

\bibliographystyle{splncs04}
\bibliography{references}

\appendix

\iflongversion
  \section{Evolution Domain Bound Axiom}

  It is often useful to assume boundedness of the evolution domain constraint.
  For reachability properties a suitable bound always exists and can be introduced syntactically with axiom \irref{diaodebound}.

  \begin{propositionE}[Evolution Domain Bound][normal]
    The axiom \irref{diaodebound} is sound:
    \[
      \cinferenceRule[diaodebound|$\abs{\&}$]{diamond ode bounding axiom}
      {\ddiamond{\pevolvein{\synodevec{\ivar}{\term}}{\fmlb}}\fml\lbisubjunct \lexists{\evolutionboundterm}\ddiamond{\pevolvein{\synodevec{\ivar}{\term}}{\fmlb\land \syntnorm{\myvect{\ivar}}<\evolutionboundterm}}\fml
        \quad
      }
      {\evolutionboundterm\notin\fml,\fmlb,\myvect{\ivar},\myvect{\term}}.
    \]
  \end{propositionE}

  \begin{proofE}
    For the forward direction, consider \(\state\in\sem{\ddiamond{\pevolvein{\synodevec{\ivar}{\term}}{\fmlb}}\fml}\) and define \(\syntvectorfield{\synodevec{\ivar}{\term}}{\state}(\myvect{\ival}) = \staterestr[{\semterm[{\staterepvarsbyvals[\state]{\myvect{\ivar}}{\myvect{\ival}}}]{\myvect{\term}}}]{\myvect{\ivar}}\) and \(\stateevolve{\state}{\mtime}=\staterepvarsbyvals[\state]{\myvect{\ivar}}{\odeflow{\syntvectorfield{\synodevec{\ivar}{\term}}{\state}}{\staterestr[\state]{\myvect{\ivar}}}{\mtime}}\) as in the definition of the semantics of continuous programs.
    There is some \(\mtime\geq 0\) such that \(\stateevolve{\state}{\mtime}\in \sem{\fml}\) and \(\stateevolve{\state}{[0,\mtime]}\in\sem{\fmlb}\).
    By compactness of \(\stateevolve{\state}{[0,\mtime]}\) there is some \(\evolutionboundtermsem>\norm[\infty]{\stateevolve{\state}{\mtimeb}}\) for all \(\mtimeb\in[0,\mtime]\).
    Hence \(\staterepvarbyval[\state]{\evolutionboundterm}{\evolutionboundtermsem}\in \sem{\ddiamond{\pevolvein{\synodevec{\ivar}{\term}}{\fmlb\land \syntnorm{\myvect{\ivar}}<\evolutionboundterm}}\fml}\).

    Soundness of the backward implication \(\leftarrow\) is immediate.
  \end{proofE}
  Note that both sides of axiom \irref{diaodebound} are \strictfml formulas if \(\fml,\fmlb\)  are strict.

  \section{Definition of Rank Function}\label{sec:rank}

  To define the order of induction, assign an ordinal rank to every \rdL formula and every program by mutual induction on formulas and programs, as follows:
  \begin{align*}
     &
    \rankof{\fml} = 0 \quad\text{for \(\fml\) a literal}
    \\
     &
    \rankof{\fml\lor\fmlb} =\rankof{\fml\land\fmlb}= \max\{\rankof{\fml},\rankof{\fmlb}\} +1
    \\
     &
    \rankof{\lexists{\ivar}{\fml}}=\rankof{\fml}+6 \qquad\rankof{\lforall{\ivar}{\fml}} =\rankof{\fml}+1
    \\
     &
    \rankof{\ddiamond{\prog}\fml} =\rankof{\dbox{\prog}\fml} = \rankof{\fml}+\rankof{\prog}+1
  \end{align*}
  \begin{align*}
    &
    \rankof{\humod{\ivar}{\term}} = 4
    &
    &
    \rankof{\pchoice{\prog}{\progb}} = \rankof{\prog}+\rankof{\progb}+2
    \\
    &
    \rankof{\ptest{\fml}} = \rankof{\fml}+1
    &
    &
    \rankof{{\prog};{\progb}} = \rankof{\progb}+1+\rankof{\prog}+1
    \\
    &
    \rankof{\pevolvein{\synodevec{\ivar}{\term}}{\fml}} = 2\cdot (\rankof{\fml}+\omega)^2
    &
    &
    \rankof{\prepeat{\prog}} = (\rankof{\prog}+1)\cdot\omega+1
  \end{align*}
  where \(\omega\) is the first limit ordinal.
  Note that since ordinal addition is not commutative, the order of addition in \(\rankof{\ddiamond{\prog}\fml} = \rankof{\fml}+\rankof{\prog}+1\) is important.
  
  The properties of \Cref{prop:orderprop} follow immediately from the definition of the rank function 
  For repetition programs \(\prepeat{\prog}\) note inductively that \[\rankof{\piter{\prog}{n}}+1\leq(\rankof{\prog}+1)\cdot n.\]

  \section{Proof of Completeness for Robust Reachability}\label{sec:appendixproof}

  \begin{proof}[Proof of \Cref{thm:maincompleteness}]\label{proof:maincompletenessproof}
    First, by repeatedly applying axioms \irref{diaassignequal}, \irref{diatest}, \irref{diachoice} and \irref{diacompose} with the duality \(\dbox{\prog}\fml\equiv\lnot\ddiamond{\prog}\lnot\fml\) in context, any \rrdL sentence is provably equivalent to a sentence which contains box modalities only in the form \(\dbox{\pevolvein{\synodevec{\ivar}{\term}}{\negrfmlb\land\syntnorm{\myvect{\ivar}}\leq\evolutionboundterm}}{\posrfml}\).
    The differential equation duality axiom \irref{odeboxdiamond} shows that any \rrdL sentence is provably equivalent to a \rrdL sentence without box modalities.

    The proof proceeds by induction on the particular class of robus sequents.
    A sequent is \emph{robust} if it is of the form \(\lsequent{\fmlsetb}{\fmlset}\), where \(\fmlset\) consists of box-free \strictfml formulas and \(\fmlsetb\) contains a formula of the form \((\term\leq\ivar\land \ivar\leq\termb)\in\fmlsetb\) for every \(\ivar\in\freevars{\fmlset}\) where \(\term,\termb\) are terms not mentioning~\(\ivar\).
    (The formula \(\ivar\in\syninterv_{\ivar}\) abbreviates \(\term\leq\ivar\land \ivar\leq\termb\).)
    The inductive step is encapsulated in the following lemma:

    \begin{lemma}\label{lem:auxlemm}
      If \(\lsequent{\fmlsetb}{\fmlset}\) is a valid robust sequent, then there are valid robust sequents \(\lsequent{\fmlsetb_1}{\fmlset_1},\ldots, \lsequent{\fmlsetb_n}{\fmlset_n}\) such that \(\fmlset_1\provsetrel\fmlset,\ldots,\fmlset_n\provsetrel\fmlset\) and
      \begin{sequentdeduction}
        \linfer
        {
          \lsequent{\fmlsetb_1}{\fmlset_1} & \ldots & \lsequent{\fmlsetb_n}{\fmlset_n}
        }
        {
          \lsequent{\fmlsetb}{\fmlset}
        }
      \end{sequentdeduction}
      is derivable in the \dL calculus.
    \end{lemma}

    The proof of this lemma is by another sub-induction.
    First, a related claim only about the potentially reducible non-basic formulas is shown.
    A formula is \emph{basic} if it does not contain any modalities, i.e., if it is a formula of real arithmetic.
    A set \(\fmlset\) of formulas is non-basic if it does not contain a basic formula.

    \begin{lemma}\label{lem:auxlemmb}
      Let \(\lsequent{\fmlsetb}{\fmlset,\fmlset'}\) be a valid robust sequent and \(\fmlset\) is non-basic.
      There are valid robust sequents \(\lsequent{\fmlsetb_1}{\fmlset_1,\fmlset'},\ldots, \lsequent{\fmlsetb_n}{\fmlset_n,\fmlset'}\) with \(\fmlset_1\provsetrel\fmlset,\ldots,\fmlset_1\provsetrel\fmlset\) and \dL derives:
      \begin{sequentdeduction}
        \linfer
        {
          \lsequent{\fmlsetb_1}{\fmlset_1,\fmlset'} & \ldots & \lsequent{\fmlsetb_n}{\fmlset_n,\fmlset'}
        }
        {
          \lsequent{\fmlsetb}{\fmlset,\fmlset'}
        }
      \end{sequentdeduction}
    \end{lemma}
    \begin{proof}[Proof of \Cref{lem:auxlemmb}]
      This claim is shown by induction on the length of \(\fmlset\) for all \(\fmlset'\).
      So suppose the claim holds for \(\fmlset\) and consider a valid robust sequent \(\lsequent{\fmlsetb}{\posrfml,\fmlset,\fmlset'}\).
      By the induction hypothesis on \(\fmlset\), it suffices to find a list of valid robust sequents \(\lsequent{\fmlsetb_1}{\fmlset_1,\fmlset,\fmlset'},\ldots, \lsequent{\fmlsetb_n}{\fmlset_n,\fmlset,\fmlset'}\) with \(\fmlset_1\provsetrel\posrfml,\ldots,\fmlset_1\provsetrel\posrfml\) such that \dL derives the following
      \begin{sequentdeduction}
        \linfer
        {
          \lsequent{\fmlsetb_1}{\fmlset_1,\fmlset,\fmlset'} & \ldots & \lsequent{\fmlsetb_n}{\fmlset_n,\fmlset,\fmlset'}
        }
        {
          \lsequent{\fmlsetb}{\posrfml,\fmlset,\fmlset'}
        }
      \end{sequentdeduction}
      To find such sequents proceed by another induction on the length of formula \(\fml\).

      \begin{caselist}
        \caseof{\(\posrfml_1\lor\posrfml_2\)}
        By applying the claim to the formulas \(\posrfml_1\) and \(\posrfml_2\),  which are shorter, note that \dL derives
        \begin{sequentdeduction}
          \linfer[orR]
          {
            \linfer
            {
              \linfer
              {
                \lsequent{\fmlsetb_1}{\fmlset_1^1,\fmlset_1^2,\fmlset,\fmlset'} & \ldots
              }
              {
                \lsequent{\fmlsetb_1}{\fmlset_1^1,\posrfml_2,\fmlset,\fmlset'}
              }
              & \ldots
            }
            {
              \lsequent{\fmlsetb}{\posrfml_1,\posrfml_2,\fmlset,\fmlset'}
            }
          }
          {
            \lsequent{\fmlsetb}{\posrfml_1\lor\posrfml_2,\fmlset,\fmlset'}
          }
        \end{sequentdeduction}
        Validity is preserved by the inductive hypothesis and \(\fmlset_i^1\provrel\posrfml_1\) and \(\fmlset_i^2\provrel\posrfml_2\) so \(\fmlset_i^1,\fmlset_i^2\provrel\posrfml_1,\posrfml_2\provrelhalf\posrfml_1\lor\posrfml_2\) by \eqcref{prop:orderprop}{it:orderdisj}.

        \caseof{\(\posrfml_1\land\posrfml_2\)}
        Letting \(\fmlsetb_1=\fmlsetb\) and \(\fmlset_1=\{\posrfml_1\}\) and \(\fmlset_2=\{\posrfml_2\}\) is sufficient by \irref{andR} and \eqcref{prop:orderprop}{it:orderconju} since both sequents are clearly valid.

        \caseof{\(\lexists{\ivar}\posrfml\)}
        Note that as \(\semopen{\posrfml}\) is open and the sequent is valid, there is a rational witness \(\ratconst\in\rationals\) for \(\ivar\).
        Choosing \(\fmlsetb_1=\fmlsetb\) and \(\fmlset_1 = \{\fmlsafereplacevarby[\posrfml]{\ivar}{\ratconst}\}\) is as required, by \irref{existsR} and \eqcref{prop:orderprop}{it:existentialred}.

        \caseof{\(\lforallbdd[\syninterv]{\ivar}\posrfml\)}
        By uniform renaming, without loss of generality, \(\ivar\) does not appear in \(\fmlsetb_1\), \(\fmlset\) or \(\fmlset'\).
        So letting \(\fmlsetb_1=\fmlset\cup\{\ivar\in\syninterv\}\) and \(\fmlset_1=\{\posrfml\}\) works by \irref{forallR} and \eqcref{prop:orderprop}{it:existentialred}.

        \caseof{\(\ddiamond{\posrprog}\)}
        Distinguish based on the top level connective of \(\posrprog\).
        For the cases of assignment \(\humod{\ivar}{\term}\), test \(\ptest{\posrfml}\), choice \(\pchoice{\posrprog}{\posrprogb}\) and composition programs \({\posrprog};{\posrprogb}\) the reduction is straightforward by axioms \irref{diaassignequal}, \irref{diatest}, \irref{diachoice}, \irref{diacompose} and properties \eqref{it:orderassign}, \eqref{it:ordertest}, \eqref{it:orderchoose}, \eqref{it:ordercompose} of \Cref{prop:orderprop}, respectively.
        Only the cases of infinite loops and differential equations remain.

        \caseof{\(\ddiamond{\prepeat{\prog}}\posrfml\)}
        Note that validity of sequent \(\lsequent{\fmlsetb}{\ddiamond{\prepeat{\posrprog}}\fml,\fmlset,\fmlset'}\) means that \(\bigcap_{\negfmlb\in\fmlsetb}\semclosed{\negfmlb}\cap\bigcap_{\posrfmlb\in\fmlset\cup\fmlset'}\complement{{\semopen{\posrfmlb}}}\subseteq \bigcup_{n=0}^\infty\semopen{\ddiamond{\piter{(\posrprog)}{n}}\fml}\).
        Note that the left hand side of this inclusion is closed as the intersection of closed sets by \Cref{thm:semantictopology}.
        Since \(\fmlsetb\) bounds all relevant variables, we can assume without loss of generality that the left hand side is compact.
        Since each \(\semopen{\ddiamond{\piter{(\posrprog)}{n}}\fml}\) is open, there is some \(n\) such that \(\bigcap_{\negfmlb\in\fmlsetb}\semclosed{\negfmlb}\cap\bigcap_{\posrfmlb\in\fmlset\cup\fmlset'}\complement{{\semopen{\posrfmlb}}}\subseteq \semopen{\ddiamond{\piter{(\posrprog)}{n}}\fml}\).
        Hence \(\lsequent{\fmlsetb}{\ddiamond{\piter{(\posrprog)}{n}}\fml,\fmlset,\fmlset'}\) is valid and choosing \(\fmlsetb_1=\fmlsetb\) and \(\fmlset_1=\{\ddiamond{\piter{(\posrprog)}{n}}\posrfml\}\) is as required by \irref{starfinite} and \eqcref{prop:orderprop}{it:repeat}.

        \caseof{\(\ddiamond{\pevolvein{\synodevec{\ivar}{\term}}{\posrfml}}\)}
        By using the Euler axiom \irref{diaode} and reducing the fresh logical connectives as in the previous cases, the sequent can be reduced to one of the form
        \begin{sequentdeduction}
          \linfer[diaodebound]
          {
            \linfer
            {
              \lsequent{\fmlsetb_,\varepsilon=0}{\ddiamond{\prepeat{(\ptest{\epsint[\varepsilon]{\posrfmlb}};\iterationprogabb)}}(\epsint[\varepsilon]{\posrfml}\land \varepsilon <1),\fmlset,\fmlset'}
            }
            {
              \vdots
            }
          }
          {
            \lsequent{\fmlsetb}{ \ddiamond{\pevolvein{\synodevec{\ivar}{\term}}{\posrfmlb}}\posrfml,\fmlset,\fmlset'}
          }
        \end{sequentdeduction}
        where \(\evolutionboundterm,\fbdd,\lipsch\) and \(h\) have been replaced by rational constants.
        The claim follows by \eqcref{prop:orderprop}{it:orderode}.
      \end{caselist}
    \end{proof}

    \begin{proof}[Proof of \Cref{lem:auxlemm}]
      Consider a valid robust sequent \(\lsequent{\fmlsetb}{\fmlset}\) and let \(\fmlset_0\) be the non-basic and \(\fmlset'\) be the basic formulas of \(\fmlset\).
      If \(\fmlset_0=\emptyset\), then \(\lsequent{\fmlsetb}{\fmlset}\) is derivable by \irref{qear} as it is a valid sequent of first-order real arithmetic.
      In other words, the empty list of robust sequents derives \(\lsequent{\fmlsetb}{\fmlset}\) and witnesses the lemma.

      If \(\fmlset_0\neq\emptyset\), by \Cref{lem:auxlemmb} there are valid robust sequents \(\lsequent{\fmlsetb_1}{\fmlset_1,\fmlset'},\ldots, \lsequent{\fmlsetb_n}{\fmlset_n,\fmlset'}\)
      from which \dL derives \(\lsequent{\fmlsetb}{\fmlset_0,\fmlset'}\) and \(\fmlset_1\provsetrel\fmlset_0,\ldots,\fmlset_1\provsetrel\fmlset_0\).
      As the rank of the basic formulas is \(0\), moreover \(\fmlset_1,\fmlset'\provsetrel\fmlset\) and \ldots \(\fmlset_n,\fmlset'\provsetrel\fmlset\).
    \end{proof}

    To conclude the proof of \Cref{thm:maincompleteness}, apply \Cref{lem:auxlemm} repeatedly, beginning with the valid robust sequent \(\lsequent{}{\fml}\) for a box-free, \strictfml \rrdL sentence.
    The result is a (possibly infinite) derivation tree for \(\prov{\fml}\).
    If this tree were infinite, there would be an infinite branch through the constructed tree by König's Lemma.
    However, every application of \Cref{lem:auxlemm} on this branch reduces the succedent of the sequent with respect to \(\provsetrel\) by construction (\Cref{lem:auxlemm}).
    Such an infinite \(\provsetrel\)-descending sequence in \(\setoffmlsets\) contradicts well-foundedness of \(\provsetrel\).
    So the derivation is finite and as such the required proof.
  \end{proof}

  \section{Additional Proofs}
  \printProofs

\fi

\end{document}